\documentclass[11pt,a4]{article}

\usepackage{enumerate}
\usepackage{latexsym} 
\usepackage{theorem}
\usepackage{graphics}
\usepackage{graphicx}
\usepackage{amsmath}
\usepackage{amssymb}
\usepackage[english]{babel} 
\usepackage{comment}
\usepackage{color}
\usepackage{tikz}
\usetikzlibrary{arrows}
\usepackage{ifpdf}

\ifpdf
\fi

\newtheorem{defeng}{Definition}[section]
\newtheorem{theorem}[defeng]{Theorem}

\newtheorem{lemma}[defeng]{Lemma}

\newtheorem{conjecture}[defeng]{Conjecture}

{\theorembodyfont{\rmfamily} }
{\theorembodyfont{\rmfamily} \newtheorem{remark}[defeng]{Remark}}
{\theorembodyfont{\rmfamily} \newtheorem{question}[defeng]{Question}}
{\theoremstyle{break}\theorembodyfont{\rmfamily} }
{\theoremstyle{break}\theorembodyfont{\rmfamily} }

\newcommand{\m}{\mathcal}
\newcommand{\mr}{\mathring}

\newcounter{claim}

\newenvironment{proof}[1][]%
 {\noindent {\setcounter{claim}{0}\sc proof ---
   }{#1}{}}{\hfill$\Box$\vspace{2ex}} 

\newenvironment{claim}[1][]%
{\refstepcounter{claim}\vspace{1ex}\noindent{(\it\arabic{claim}){#1}{}}\it}{\vspace{1ex}}

\newenvironment{proofclaim}[1][]%
	{\noindent {}{#1}{}}{ This proves~(\arabic{claim}).\vspace{1ex}}

\newcommand{\sm}{\setminus} 

\usepackage{ifpdf}

\ifpdf
\fi

\title{Excluding cycles with a fixed number of chords}

\author{Pierre Aboulker\thanks{LIAFA, Universit\'e Paris 7 --
     Paris Diderot (France),\ email: pierreaboulker@liafa.jussieu.fr.
    Partially supported by \emph{Agence Nationale de la Recherche} under reference
    \textsc{anr 10 jcjc 0204 01}.}
    \and Nicolas Bouquet\thanks{LIRMM, Universit\'e Montpellier II (France), \ email:bousquet@lirmm.fr}
}
    
\begin{document}

\maketitle

\begin{abstract}
Trotignon and Vu\v{s}kovi\'c completely characterized graphs that do not contain cycles with exactly one chord. In particular, they show that such a graph $G$ has chromatic number at most $\max(3,\omega(G))$. We generalize this result to the class of  graphs that do not contain cycles with exactly two chords and the class of graphs that do not contain cycles with exactly three chords. \\
More precisely we prove that graphs with no cycle with exactly two chords have chromatic number at most $6$. And a graph $G$ with no cycle with exactly three chords have chromatic number at most $\max(96,\omega(G)+1)$.
\end{abstract}

\section{Introduction}

The \emph{chromatic number} of a graph $G$, denoted by $\chi(G)$, is the minimum number of colors needed to vertex-color $G$ such that two adjacent vertices receive distinct colors. 
A clique is a graph such that every two vertices are adjacent. 
The \emph{clique number} of a graph $G$, denoted by $\omega(G)$, is the number of vertices of the largest clique in $G$.
A class of graphs is \emph{hereditary} if, for any graph $G$ in the class, every subgraph of $G$ is in the class.
We say that a graph $G$ is {\em $H$-free} if  $G$ does not contain the graph $H$ as an induced subgraph. 
If $\mathcal H$ is a class of graphs we say that a graph $G$ is $\mathcal H$-free if for any $H \in \mathcal H$, $G$ is $H$-free. 
Classes of graphs defined by forbidding some graphs as induced subgraphs are clearly hereditary.

It is clear that $\omega(G)$ is  a lower bound of $\chi(G)$ since vertices of a clique are colored with pairwise distinct colors. 
Gy\'arf\'as introduced the following notion \cite{gyarfas:perfect}: a graph is said to be \emph{$\chi$-bounded} if there exists a function $f$ such that for every subgraph $H$ of $G$, $\chi(H) \leq f(\omega(H))$. 
A class of graphs $\m C$ is said to be \emph{$\chi$-bounded} is every graph in the class is $\chi$-bounded. 
Observe that, in order to prove that a hereditary class $\m C$ is $\chi$-bounded by a function $f$, it is enough to prove that for any graph $G$ in $\m C$, $\chi(G) \le f(\omega(G))$. 
For instance, graphs $\chi$-bounded by the function $f(x)=x$ are known as \emph{perfect graphs}. M.~Chudnovsky, N.~Robertson, P.~Seymour, and R.~Thomas proved \cite{chudnovsky.r.s.t:spgt} that perfect graphs are exactly the graphs that do not admit odd cycles of length at least $5$ nor complement of odd cycles of length at least $5$, solving the famous strong perfect graph conjecture proposed by C.~Berge \cite{perfectConj}.
So, a natural question arises:

\begin{question}
What kind of induced structure is needed to be forbidden in order to get a $\chi$-bounded class?
\end{question}

Let us now survey some results on $\chi$-boundedness by emphasizing on what different meanings ``structure'' can take.

\medskip

If $H$ is a graph, we denote by \textit{Forb(H)} the class of $H$-free graphs.
A first way to tackle the problem is to determine for which graphs $H$, Forb($H$) is $\chi$-bounded. 
For example, it is proved in \cite{gyarfas:perfect} that Forb($P_k$) is $\chi$-bounded (where $P_k$ denotes the chordless path of length $k$).
In  \cite{erdos}, Erd\H{o}s proved that there exists graphs  with arbitrarily large chromatic number and arbitrarily large girth. 
So, if $H$ contains a cycle, Forb($H$) is not $\chi$-bounded.
It is actually conjectured in \cite{gyarfas:perfect} that Forb($H$) is $\chi$-bounded if and only if $H$ is a forest. 
The deeper results concerning this conjecture are certainly  results of Kierstead and Penrice \cite{radius2} and Kierstead and Zhu~\cite{radius3} proving that the conjecture holds for every tree of radius at most $2$ and several trees of radius $3$.
To get out from this conjecture, we need to forbid a class of graph $\mathcal H$ such that $\mathcal H$ contains graphs with arbitrarily large girth.

\medskip

A second way to forbid induced structures is the following: fix a graph $H$, and forbid every induced subdivision of $H$. 
We denote by \textit{Forb*($H$)} the class of graphs that do not contain induced subdivisions of $H$.
The class Forb*($H$) has been proved to be $\chi$-bounded for a number of examples.
The most beautiful one is certainly the proof by Scott \cite{scottTree} that for any forest $F$, Forb*($F$) is $\chi$-bounded.  
In the same paper he  conjectured that, for any graph $H$, Forb*($H$)  is $\chi$-bounded. 
Unfortunately, this conjecture has recently been disproved by Kozik \textit{et al.} \cite{polack}. 
Based on this work, Chalopin \textit{et al.}~\cite{liScott} gave a precise description of a number of graphs $H$ for which Forb*($H$) is not $\chi$-bounded.
There is no general conjecture on which $H$ has to be forbidden in order to ensure Forb*($H$) is $\chi$-bounded.

\medskip

A third way is to forbid a graph $H$ for which some edges can be subdivided but some cannot. 
More generally, to forbid a class of graphs $\m H$ such that, for each $H \in \m H$, some edges can be subdivided and some cannot. 
A few classes defined this way has  been studied (see \cite{aboulkerRTV:propeller} and \cite{nicolas.kristina:one} for instance).
In \cite{nicolas.kristina:one}, Trotignon and Vu\v{s}kovi\'c proved that the class of graphs that do not contain cycles with a unique chord is $\chi$-bounded by the function $max(3, \omega(G))$. 
Forbidding  cycles with a unique chord is equivalent to forbid a diamond (a  diamond is a cycle of length $4$ with a diagonal) such that every edge but the diagonal can be subdivided.

A \emph{k-cycle} is a chordless cycle with exactly $k$ chords. 
We call $\mathcal C_k$ the class of $k$-cycle-free graphs i.e.\ the class of graphs that do not contain cycles with exactly $k$ chords. 
So, the cited result on the class of graphs that do not contain a cycle with a unique chord may be rephrased as follows : $\mathcal C_1$ is $\chi$-bounded.

\medskip

The two main results of this paper are that both $\mathcal C_2$ (see Theorem \ref{2-cycle}) and $\mathcal C_3$ (see Theorem \ref{3-cycleclique4}) are $\chi$-bounded.
The statement of Theorem \ref{2-cycle} deals with a super-class of $\mathcal C_2$, see Section \ref{sec2-cycle} for more details. Since graphs which do not contain a $2$-cycle do not contain $K_4$ as an induced subgraph, proving $\chi$-boundedness is equivalent to prove that the chromatic number is bounded. We actually prove that the chromatic number in this case is at most $6$. An immediate lower bound on the chromatic number is $3$, given by odd cycles. Close the gap between these two values can be an interesting point.

The class $\m C_3$, contrary to $\m C_2$ that does not admit graphs with cliques larger than the triangle (because $K_4$ is a $2$-cycle), admits graphs containing arbitrarily large cliques. We prove that the chromatic number of a graph $G \in \m C_3$ is at most $\max(96,\omega(G)+1)$. 
In addition we provide examples of graphs $G$ with arbitrarily large clique such that $\chi(G) = \omega(G)+1$, showing that our bound is asymptotically tight.
Nevertheless, lower bound of $96$ is surely far away from an optimal bound for graphs in $\m C_3$ that do not contain large cliques.

\medskip

Here is an outline of the paper.
In Section \ref{notations}, we give some terminologies and  in Section \ref{secPrelCycle} we describe the general method used in the proofs. Section \ref{sec2-cycle} is concerned with the class $\m C_2$ and Section \ref{sec3-cycle}  is concerned with the class $\m C_3$. 
We also propose the following conjecture suggested by our results: 

\begin{conjecture}
For any integer $k \ge 4$, $\mathcal C_k$ is $\chi$-bounded.
\end{conjecture}

\section{Terminologies and notations} \label{notations}

For standard definition on graphs, the reader should refer to classical books of graph theory, such as~\cite{diestel}.
Let $G$ be a graph, $x$ a vertex of $G$ and $S$ a subset of vertices of $G$.
We denote by $N(x)$ the set of neighbors of~$x$, by $N_S(x)$ the set of neighbors of $x$ in $S$, and by $N(S)$ the set of vertices of $V(G)\sm S$ that have a neighbor in $S$. We denote by $d(x)$ the degree of $x$ and by $d_S(x)$ the degree of $x$ in $S$, i.e.~the number of neighbors of $x$ in $S$.
We denote by $G[S]$ the subgraph of $G$ induced by $S$, and $G\setminus S$ denotes $G[V(G)\setminus S]$.
$S$ is a {\em cutset} of $G$ if
$G \setminus S$ has more than one connected component.
If $S$ induces a clique, then $S$ is a \textit{clique cutset}.
If $\{ x \}$ is a cutset of $G$, then $x$ is a {\em cutvertex}. Note that a cutvertex is a clique cutset.

A {\em path} $P$ is a sequence of distinct vertices $p_1p_2\ldots p_k$,
$k\geq 1$, such that $p_ip_{i+1}$ is an edge for every $1\leq i <k$.
For every $1\leq i <k$, the edge $p_ip_{i+1}$ is an {\em edges of $P$}.
Vertices $p_1$ and $p_k$ are the {\em endpoints} of $P$, and 
$p_2\ldots p_{k-1}$ is the {\em interior} of $P$.
$P$ is referred to as a {\em $p_1p_k$-path}. 
For $1\le i \le j \le k$, we write $p_jPp_i:=p_j \ldots p_i$, $\mathring P:=p_2 \dots p_{k-1}$, $\mathring p_j P\mathring p_i:=p_{j+1} \dots p_{i-1}$.

A cycle $C$ is a sequence of vertices $p_1p_2\ldots p_kp_1$, $k \geq 3$, such that $p_1\ldots p_k$ is a path and $p_1p_k$ is an edge.
Edges $p_ip_{i+1}$,  for  $1\leq i <k$, and edge $p_1p_k$ are called the {\em edges of $C$}.
Let $Q$ be  a path or a cycle. 
The {\em length} of $Q$ is the number of its edges.
An edge $e=uv$ of $G$ is a {\em chord} of $Q$ if $u,v\in V(Q)$, but $uv$ is
not an edge of $Q$. 
A path or a cycle $Q$ in a graph $G$ is {\em chordless} if no edge of $G$ is a chord of $Q$.



A graph $G$ is a {\em complete k-partite} if $V(G)$ can be partitioned into $k$ non-empty subsets $A_1, \dots, A_k$ such that, for $i=1, \dots, k$, $A_i$ is a stable set and, for any $\{i,j\} \subseteq \{1, \dots, k\}$, there is all possible edges between $A_i$ and $A_j$. Sets $A_i$ are called the \emph{partitions' set} of $G$. $G$ is noted $K_{a_1, \dots, a_k}$ where $a_i=|A_i|$ for $i=1, \dots, k$.
 If $k=2$ then $G$ is said to be a {\em complete bipartite graph} and if $k=3$, $G$ is said to be a {\em complete tripartite graph}. 
The graph  $K_{1,1,2}$  is called a {\em diamond}.

\section{Preliminaries} \label{secPrelCycle}

We mentioned that $\m C_1$ was already proved to be $\chi$-bounded, we use this result for graphs in $\m C_1$ that contain no $K_4$, which formally give:

\begin{theorem}[Trotignon and Vu\v{s}kovi\'c \cite{nicolas.kristina:one}] \label{1-cycle}
If $G \in \m C_1$ and $\omega(G) \le 3$, then $\chi(G) \le 3$.
\end{theorem}

Let us now explain a classical tool to prove $\chi$-boundedness results for classes of graphs defined by forbidding induced structure and that is extensively use in this paper.

Let $G$ be a graph.
The \textit{distance} between two vertices $x$, $y$ of $G$ is the length of a shortest $xy$-path. 
Let $z$ be a vertex of $G$ and let $i$ be an integer. 
The \emph{$i$-th level} of $z$ is the set of vertices, denoted by $S_i(z,G)$, that are at distance exactly $i$ from $z$ in $G$. 
If no confusion is possible, we denote it by $S_i(z)$ in order to avoid too heavy notations. 
A \emph{father} of a vertex $x \in S_i(z)$ is a vertex in $S_{i-1}(z)$ adjacent to $x$. 
For every pair of vertices $x$, $y$ in $S_{i}(z)$, it is easy to see that there exists a chordless $xy$-path $Q$ with internal vertices included in $z \cup S_1(z) \cup \dots \cup S_{i-1}(z)$ such that only the endpoints of  $\mathring{Q}$ are in $S_{i-1}(z)$.
Note that, as a consequence, the endpoints of $\mathring{Q}$ are the only vertices that have neighbors in $S_i(z)$.
Such paths are called \textit{unimodal paths}
and are a key tool to find particular induced structures in a graph.
In the remaining of the paper, the letter $Q$ is reserved to denote unimodal paths and the following convention is followed: if $x$ and $y$ are two vertices in $S_i(z)$, $Q_{xy}$ denotes a unimodal path with endvertices $x$ and $y$.

The following general remark explains the reason why decomposing a graph into levels (as described above) is a very powerful tool to bound its chromatic number.

\begin{remark} [Folklore]\label{scott}
Let $G$ be a graph  and let $z$ be a vertex of $G$. 
There exists an integer $k$ such that $G[S_k(z)]$ has chromatic number at least $\lceil \chi(G) /2 \rceil$.
\end{remark}

\begin{proof}
For any $i \neq j$, there is no edges between a vertex of $S_{2i}$ (resp. of $S_{2i+1}$) and of $S_{2j}$ (resp. of $S_{2j+1}$). Indeed adjacent vertices are at distance one, so their level differ by at most one.
So, $$\chi(G) \le max_{i\ even} \chi(G[S_i]) + max_{j\ odd} \chi(G[S_j]).$$
The result follows.

\end{proof}



\section[Exactly two chords]{Graphs that do not contain a cycle with exactly two chords as induced subgraph} \label{sec2-cycle}

Let $C$ be a cycle with exactly two chords $e_1=a_1a_2$ and $e_2=b_1b_2$.
If $e_1$ and $e_2$ share an extremity, then $e_1$ and $e_2$ are said to be {\em V-chords} of $C$.
If $a_1$, $a_2$, $b_1$, $b_2$ are pairwise distinct and appear in the following order along $C$ : $a_i, b_j,a_k,b_l$ with $\{i,k\}=\{j,l\}=\{1,2\}$, then $e_1$ and $e_2$ are said to be {\em crossing chords} of $C$. 
A 2-cycle with V-chords (resp.\ crossing chords) is called a \emph{V-cycle} (resp.\ an \emph{X-cycle}) (see Figure \ref{X-cycleV-cycle}). 
A 2-cycle that is not a V-cycle nor an X-cycle is a \textit{parallel cycle}.

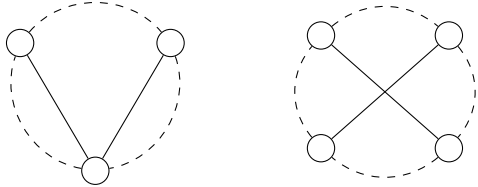
\begin{figure}[h]

\centering
\begin{tikzpicture} 

\node [draw,circle] (x) at (2,0) {};
\node [draw,circle] (x1) at (1,1.7) {};
\node [draw,circle] (x2) at (3,1.7) {};

\draw [dashed] (x) to[bend left=50] (x1);
\draw [dashed] (x1) to[bend left=50] (x2);
\draw [dashed] (x2) to[bend left=50] (x);

\draw  (x)-- (x1);
\draw  (x)-- (x2);


\node [draw,circle] (y1) at (5,1.8) {};
\node [draw,circle] (z1) at (6.7,1.8) {};
\node [draw,circle] (y2) at (6.7,0.3) {};
\node [draw,circle] (z2) at (5,0.3) {};

\draw [dashed] (y1) to[bend left=40] (z1);
\draw [dashed] (z1) to[bend left=40] (y2);
\draw [dashed] (y2) to[bend left=40] (z2);
\draw [dashed] (z2) to[bend left=40] (y1);

\draw  (z1)-- (z2);
\draw  (y1)-- (y2);

\end{tikzpicture}
 \caption{A $V$-cycle and an $X$-cycle.}
\label{X-cycleV-cycle}
\end{figure}

\vspace{2ex}

The main result of this section is concerned with the class of (X-cycle, V-cycle)-free graphs that of course strictly contains the class of graphs with no 2-cycles. Indeed if a graph is (X-cycle, V-cycle)-free, it can contains a $2$-cycle if the two chords of the cycle are parallel.

\begin{theorem} \label{2-cycle}
Every (X-cycle, V-cycle)-free graph $G$ satisfies $\chi(G) \le 6$.
\end{theorem}

Note that we can bound the chromatic number by a constant here because $K_4$ is an X-cycle. 
The proof is built on two steps: Lemmas \ref{nodiamond} and \ref{2-cycleStep2}.

The first step consists in showing that we can decompose an (X-cycle, V-cycle)-free graph around a complete multipartite graph. More precisely we prove that  if a  graph is  (X-cycle, V-cycle)-free then either it has a clique cutset, or it is  a complete tripartite graph, or it is diamond-free.
In the second part, we prove that (diamond, V-cycle, X-cycle)-free graphs have chromatic number at most $6$ using Remark \ref{scott}. 
The proof is somehow based on an induction on the number of chords. The induction is based on the result of Trotignon and  Vu\v{s}kovi\'c about $\m C_1$ (Lemma~\ref{1-cycle}). 
We finally combine these two lemmas to prove Theorem~\ref{2-cycle}.

\medskip

Note that the first step actually gives us a decomposition theorem for (X-cycle, V-cycle)-free graphs, where the basic classes are complete multipartite graphs and (diamond, X-cycle, V-cycle)-free graphs. 
Anyway, to get a usable decomposition theorem, one should decompose (diamond, X-cycle, V-cycle)-free graphs that is a too complex class to be a basic class.

\begin{lemma} \label{nodiamond}
If  $G$ is an (X-cycle, V-cycle)-free graph, then either $G$ has a clique cutset, or $G$ is isomorphic to a complete tripartite graph, or $G$ is diamond-free.
\end{lemma}

\begin{proof}
Assume by way of contradiction that $G$ does not admit a clique cutset, $G$ is not isomorphic to a complete tripartite graph and $G$ contains a diamond. 
Since $G$ has no clique cutsets, $G$ is 2-connected.
Let $H=K_{i,j,k}$ be a maximum (subject to its number of vertices) complete tripartite subgraph of $G$. 
Note $A=\{a_1, \dots, a_i\}$ (resp. $B=\{b_1, \dots, b_j\}$, resp. $C=\{c_1, \dots, c_k\}$) the set of the partition of $H$ of cardinality $i$ (resp. $j$, resp. $k$).
Note that since  $G$ contains a diamond, one of the integers $i$, $j$, $k$ is strictly greater than $1$. 
Moreover, since $G$ is $K_4$-free, $H$ is an induced subgraph of $G$  i.e.\ $A$, $B$ and $C$ are stable sets. 

\begin{claim} \label{diamonddeg1}
A vertex $u \notin V(H)$ has at most one neighbor in $H$.
\end{claim}

\begin{proofclaim}
Assume by way of contradiction that some vertex $u \notin V(H)$ satisfies $d_H(u) \geq 2$.
If $u$ has a neighbor in $A$, $B$ and $C$, say $a_1$, $b_1$ and $c_1$, then $u a_1 b_1 c_1$ is a $K_4$, a contradiction.
So we may assume w.l.o.g.\ that $u$ does not have any neighbor in $C$.
Assume that $u$ has a neighbor  in $A$ and a neighbor in $B$, say $a_1$ and $b_1$. 
By maximality of $H$, $u$ has at least one non-neighbor in $A \cup B$. 
Assume w.l.o.g. that $a_2$ is a non-neighbor of $u$.
Then $ua_1c_1a_2b_1u$ is a V-cycle with chords $b_1a_1$ and $b_1c_1$, a contradiction.
So we may assume w.l.o.g. that $u$ does not have any neighbor in $B$ and thus have at least two neighbors in $A$, say $a_1$ and $a_2$. 
Then $ua_1b_1c_1a_2u$ is an X-cycle with chords $a_1c_1$ and $a_2b_2$, a contradiction.
\end{proofclaim}

Note that $G \neq H$ since otherwise $G$ is a complete tripartite graph. Let $K$ be a connected component of $G \sm H$. 
By (\ref{diamonddeg1}), vertices of $K$ that have a neighbor in $H$, have a unique neighbor in $H$.
Since $G$ does not contain clique cutsets, $N_H(K)$ must contain two non-adjacent vertices. 
Therefore, $K$ contains a chordless path $P=p_1 \dots p_k$ such that the neighbors of $p_1$ and $p_k$ in $H$ are two non-adjacent vertices. 
Among all such paths, let $P$ be minimal.  
Assume w.l.o.g.\ that $a_1$ and $a_2$ are the neighbors of respectively $p_1$ and $p_k$ in $H$. 
By minimality of $P$, no interior vertex of $P$ has a neighbor in $A$.

If  no interior vertex of $P$ is adjacent to a vertex in $B$ or $C$, then $a_1Pa_2b_1c_1a_1$ is an X-cycle with chords $a_1b_1$ and $a_2c_1$, a contradiction. 
Let $i$ be the smallest integer such that $p_i$ has a neighbor in $B$ or $C$, say $p_i$ is adjacent to $b_1$. 
Then no interior vertices of $p_1 \dots p_i$ is adjacent to any vertices of $H$ and thus $a_1p_1Pp_ib_1a_2c_1a_1$ is a V-cycle with chords $b_1a_1$ and $b_1c_1$, a contradiction.
\end{proof}


\begin{lemma} \label{2-cycleStep2}
If $G$ is a (diamond, X-cycle, V-cycle)-free graph, then for any $z \in V(G)$ and for every integer $k$, $S_k(z) \in \m C_1$.
\end{lemma}

\begin{proof}
Let $G$ be a (diamond, X-cycle, V-cycle)-free graph and let $z \in V(G)$.
Assume by way of contradiction that there exists an integer $k$ such that $S_k(z)$ contains a 1-cycle $C$ as an induced subgraph.
Name $a$, $b$ the extremities of the unique chord of $C$. 
The cycle $C$ is edge-wise partitioned in two $ab$-path: $P^l$ and $P^r$  (for left and right path, see Figure \ref{1chord}).

\begin{figure} [h]
\centering
\begin{tikzpicture} 


\draw (1 , 0) node  [below] {$a$};
\draw (1 , 2) node [above] {$b$};

\draw (1 , 2) node {$\bullet$};
\draw (1 , 0) node {$\bullet$};


\draw   [dashed] (1,0)  to[bend left=40]  (0,1)  to[bend left=40]  (1 , 2);
\draw   [dashed] (1,2)  to[bend left=40]  (2,1)  to[bend left=40]  (1 , 0);

\draw (1 , 0) -- (1 , 2);

\draw (0 , 1) node[left]  {$P^l$};
\draw (2 , 1) node[right]  {$P^r$};

\draw (1,1) ellipse (-3.5 and 1.7);
\draw (-2 , 2) node[left]  {$S_k(z,G)$};

\end{tikzpicture}

\caption{In $S_(z,G)$, the  cycle $C$ with a unique chord $ab$ and the paths $P^r$ and $P^l$}   \label{1chord}
\end{figure}
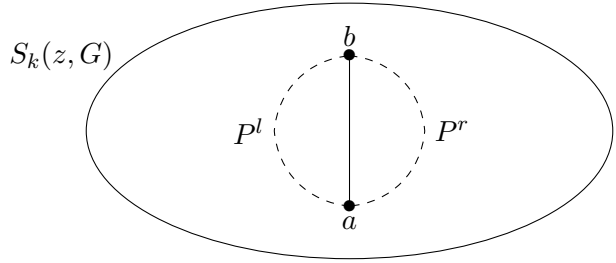

Observe that, since $G$ is V-cycle-free, no vertex of $G$ has four neighbors on an induced path. 
Moreover, no vertex $x$ of $G$ has four neighbors on an induced cycle. 
Indeed, either two consecutive neighbors are non adjacent, and then $x$ has four neighbors on an induced path, or the cycle is a $C_4$ and then, $G$ contains a diamond.

\begin{claim}\label{degreefather2}
For any $x \notin V(C)$, $d_C(x) \le 3$ and, if $x$ is adjacent to $a$ or $b$, then $d_C(x) \le 2$.
\end{claim}

\begin{proofclaim}
Let $x \notin V(C)$.

Suppose first that $x$ is adjacent to $a$ and that $d_C(x) \geq 3$.
First assume that $xb$ is an edge. 
Since $d_C(x)\geq 3$, $x$ has another neighbor $x_1$ in $C$. We can assume w.l.o.g. that $x_1$ is on $P^l$. 
If $x$ has no other neighbors on $P^l$, then $axbP^la$ is an X-cycle with chords $xx_1$ and $ab$. 
So $x$ has another neighbor $x_2$ on $P^l$ and then it has four neighbors in the chordless cycle $a P^l b a$, a contradiction.

So $xb$ is not an edge. 
Since $C \backslash \{b\}$ is an induced path, $d_C(x)=3$. 
Denote by $x_1$ and $x_2$ the two other neighbors of $x$ on $C$ distinct from $a$.
If $x_1$, $x_2$ are on $P^l$, then $aP^lx_1xx_2P^lbP^ra$ (resp. $aP^lx_1xx_2P^lba$) is a V-cycle (resp. an X-cycle) if $x_1x_2$ is not an edge (resp.~is an edge) on $a$ (resp.~with chords $ax$ and $x_1x_2$), a contradiction. 
Hence, by symmetry, we may assume that $x_1\in P^l$ and $x_2 \in P^r$. 
Since $G$ is diamond-free, $a$ is not adjacent to both $x_1$ and $x_2$. 
Assume w.l.o.g.\ $x_1a$ is not an edge. 
Thus $x_1xaP^rbP^lx_1$ is an X-cycle with chords $xx_2$ and $ab$, a contradiction.

So the second outcome of the claim holds. 
Now, if $x$ has at least 4 neighbors in $C$, then $x$ is adjacent neither to $a$ nor to $b$, and thus it has four neighbors on an chordless path, a contradiction. 
\end{proofclaim}

\begin{claim}\label{commun}
Vertices $a$ and $b$ do not have a common father.
\end{claim}

\begin{proofclaim}
Recall that, given two vertices $x,\, y$ in $S_{k-1}(z)$, $Q_{xy}$ denotes a unimodal path from $x$ to $y$. 
And interior vertices of $Q_{xy}$ are not adjacent to any vertex of $C$.

Assume by way of contradiction that there exists a vertex $x \in S_{k-1}(z)$ that is a common father to $a$ and $b$.
Let $c$ be the neighbor of $a$ on $P^r$ and $d$ be a father of $c$. 
By (\ref{degreefather2}), $d \neq x$ and, since $G$ is diamond-free, $P^l$ and $P^r$ have length at least $3$ i.e.\ $bc$ is not an edge. 

If $d$ is adjacent to $a$ then $c d Q_{dx} x b a c$ is a V-cycle on $a$.
If $d$ is adjacent to $b$ then $c d Q_{dx} b a c$ is an X-cycle with chords $ax$ and $bd$.
So $d$ is adjacent neither to $a$ nor to $b$.

Vertex $d$ has at least one neighbor $d_1$ on $\mathring P^l$, otherwise $cdQ_{dx}bP^lac$ is a V-cycle on $a$. 
Moreover, $d$ has a neighbor $d_2 \neq c$ on $\mathring P^r$ otherwise $cdQ_{dx}abP^rc$ is an X-cycle with chords $ac$ and $xb$. 
By Claim~\ref{degreefather2}, $d_C(x) \leq 3$, so $N_C(d)=\{c,d_1,d_2\}$.

If $d_1b$ is not an edge, then $d_1dcP^rbaP^ld_1$ is an X-cycle with chords $ac$ and $dd_2$. 
Otherwise $abxQ_{xd}d_1P^la$ is an X-cycle with chords $bd_1$ and $ax$, a contradiction.
\end{proofclaim}

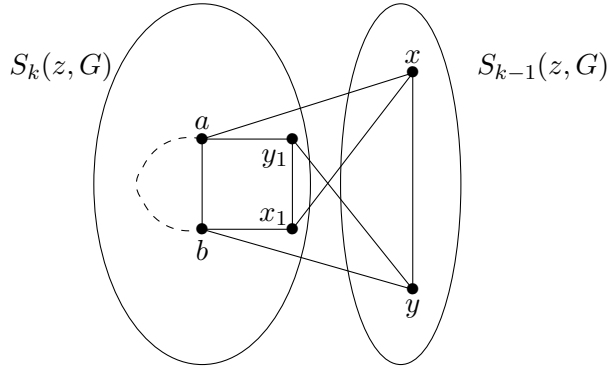
\begin{figure}[h]
\centering
\begin{tikzpicture} [scale=0.8]


\draw (0 , 1.5) node  [above] {$a$};
\draw (0 , 0) node [below] {$b$};
\draw (1.6 , 1.5) node [below left] {$y_1$};
\draw (1.6 , -0.1) node [above left] {$x_1$};

\draw (0 , 1.5) node  {$\bullet$};
\draw (0 , 0) node  {$\bullet$};
\draw (1.5 , 1.5) node  {$\bullet$};
\draw (1.5 , 0) node  {$\bullet$};


\draw (0,0) -- (0,1.5) -- (1.5,1.5) -- (1.5,0) -- (0,0);

\draw (3.5 , 2.6) node  [above] {$x$};
\draw (3.5 , -1) node [below] {$y$};

\draw (3.5 , 2.6) node  {$\bullet$};
\draw (3.5 , -1) node  {$\bullet$};

\draw (0,1.5) -- (3.5 ,2.6);
\draw (1.5,0) -- (3.5 ,2.6);

\draw (0,0) -- (3.5 ,-1);
\draw (1.5,1.5) -- (3.5 ,-1);

\draw (3.5,2.6) -- (3.5 ,-1);


\draw   [dashed] (0,1.5)  to[bend right=40]  (-1 , 1) to[bend right=40]  (-1 , 0.5)  to[bend right=40]  (0,0);


\draw (0,0.75) ellipse (1.8 and 3);
\draw (-1.3 , 2.7) node[left]  {$S_k(z,G)$};

\draw (3.3 ,  0.75) ellipse (1 and 3);
\draw (4.4 , 2.7) node[right]  {$S_{k-1}(z,G)$};

\end{tikzpicture}
\caption{The particular graph of Claim \ref{fatherdegree1}, note that $P^r=ay_1x_1b$.}
\label{particulargraph}
\end{figure}

\begin{claim}\label{fatherdegree1}
 Both $a$ and $b$ have a father of degree $1$.
\end{claim}

\begin{proofclaim}
Let $x$ and $y$ be respectively a father of $a$ and a father of $b$. 
By (\ref{commun}), $x \neq y$.
Assume by way of contradiction that say $x$ has a neighbor  $x_1 \neq x$ in $\mathring{P^r}$.
By (\ref{degreefather2}), $N_C(x)=\{a,x_1\}$.
If $y$ has no neighbor in $\mathring{P}^r$, then $axQ_{xy}ybP^ra$ is an X-cycle with chords $ab$ and $xx_1$, a contradiction.
So $y$ has a neighbor, say $y_1$, in $\mathring{P}^r$ and, by (\ref{degreefather2}), $N_C(y)=\{b,y_1\}$.

Suppose first that $x_1=y_1$.
Since $G$ is diamond-free, $bx_1$ and $ax_1$ cannot both be edges. So we may assume w.l.o.g.\ that $bx_1$ is not an edge. 
Then, $aP^rx_1xQ_{xy}yba$ is an X-cycle with chords $ax$ and $x_1y$, a contradiction. So $x_1 \neq y_1$.

Now, if $a$, $x_1$, $y_1$ appear in this order along $P^r$, then $axQ_{xy}y_1P^rbP^la$ is a V-cycle on $b$, a contradiction. 
So $a$, $y_1$, $x_1$ appear in this order along $P^r$.
If $ay_1$ is not an edge, then $axQ_{xy}yy_1P^rba$ is an X-cycle with chords $xx_1$ and $by$, a contradiction.
So $ay_1$ is an edge and, symmetrically, $bx_1$ is an edge.
If $y_1x_1$ is not an edge, then $ay_1yQ_{yx}x_1ba$ is an X-cycle with chords $ax$ and $by$, a contradiction.
So $x_1y_1$ is an edge (i.e.\ $ax_1y_1x_1b$ is a square). 
If $xy$ is not an edge, then $abyy_1x_1xa$ is an X-cycle. 
So $xy$ is an edge (see Figure \ref{particulargraph}).

Let $c$ be the neighbor of $a$ on $P_l$ and let $d$ a father of $c$. 
First assume that $cb$ is an edge. Since $G$ is diamond-free, $d$ is adjacent neither to $a$ nor to $b$.
If $x_1d$ is not an edge then $cdQ_{dx}xx_1bac$ is a X-cycle with chords $cb$ and $xa$, a contradiction.
Hence, by symmetry, $d$ is adjacent to both $x_1$ and $y_1$ and then $axQ_{xd}dx_1y_1a$ is an X-cycle with chords $dy_1$ and $xx_1$.

Hence $cb$ is not an edge. 
If $d$ is adjacent to both $x_1$ and $y_1$, then $axQ_{xd}dx_1y_1a$ is an X-cycle with chords $dy_1$ and $xx_1$.
If $d$ is adjacent neither to $x_1$ nor to $y_1$, $cdQ_{dy}ybx_1y_1ac$ is a X-cycle with chords $ab$ and $yy_1$. 
If $d$ is adjacent to $x_1$ and not to $y_1$ then $cdQ_{dx}xx_1bac$ is a X-cycle with chords $ax$ and $dx_1$. 
If $d$ is adjacent to $y_1$ and not to $x_1$ then $cdQ_{dx}xx_1y_1ac$ is a X-cycle with chords $yy_1$ and $xa$
\end{proofclaim}

By (\ref{fatherdegree1}), there exist two vertices $x$ and $y$ such that $x$ is a father of $a$, $y$ is a father of $b$ and $d_C(x)=d_C(y)=1$. 
Since $G$ is diamond-free, $P^r$ and $P^l$ cannot be both of length two, so we may assume w.l.o.g.\ that $P^l$ has length at least 3.
Let $c$ be the neighbor of $a$ on $P^l$ and $d$ be a father of $c$. 
Note that $d \neq x$ and $d \neq y$. 
If $d_C(d)=1$, then $axQ_{xd}dcP^lbP^ra$ is a V-cycle on $a$, a contradiction. 
Hence $d_C(d) \ge 2$.

Assume first $d$ has a neighbor $d_1$ in $P^r$.
If $d_1$ is the unique neighbor of $d$ in $P^r$, then $cdQ_{dy}ybP^rac$ is an X-cycle with chords $ab$ and $dd_1$ if $a \neq d_1$, or  is a V-cycle on $a$ if $d=a$ or $d=b$.
So $d$ has a second neighbor $d_2$ in $P^r$ and thus, by (\ref{degreefather2}), $N_C(d)=\{c,d_1,d_2\}$ and $\{d_1,d_2\} \subseteq \mathring P^r$.
Assume w.l.o.g.\ that $a$, $d_1$ and $d_2$ appear in this order along $P^r$.
Now, $axQ_{xd}dd_2P^rbP^la$ is an X-cycle with chords $cd$ and $ab$, a contradiction.
So $d$ has no neighbors in $P^r$ and thus has some neighbors in $\mathring P^l$.

If $d$ has exactly one neighbor  $d_1$ in $\mathring P^l$, then $axQ_{xd}dcP^lba$ is an X-cycle with chords $dd_1$ and $ac$, a contradiction.
So $d$ has at least two neighbors in $\mathring P^l$ and, by (\ref{degreefather2}), it has exactly two.
Put $N_C(d)=\{c,d_1,d_2\}$.
Now, $cdQ_{dy}ybP^lc$ is a V-cycle with chords $dd_1$ and $dd_2$, a contradiction.
\end{proof}

We are now ready to prove Theorem \ref{2-cycle}, recall that this theorem states that \textit{every (X-cycle, V-cycle)-free graph is  6-colorable.}

\medskip

\begin{proof} 
Assume by way of contradiction that there is some (X-cycle, V-cycle)-free graphs that are not $6$-colorable.
Let $G$ be  minimal with this property.
Suppose first that $G$ contains a diamond. 
Since a complete tripartite graph is 3-colourable, by Lemma \ref{nodiamond}, G admits a clique cutset $K$. 
Let $C_1$ be a connected component of $G \sm K$, and $C_2$ the union of all other components of $G \sm K$. 
Put $G_1 = G[C_1 \cup K)]$ and $G_2 = G[C_2 \cup K)]$. 
If $G_1$ and $G_2$ are both 6-colourable, then $G$ is 6-colourable, a contradiction. 
Therefore $G_1$ or $G_2$ is not 6-colourable, a contradiction to the minimality of $G$. 
So we may assume that $G$ is diamond-free i.e., $G$ is (diamond, X-cycle, V-cycle)-free. 

Let $z$ be a vertex of $G$.
Since $\chi(G) =7$, by Remark \ref{scott}, there is an integer $k$ such that $\chi(S_k(z)) \ge 4$. 
So, by Theorem \ref{1-cycle}, $S_k(z)$ contains a 1-cycle as an induced subgraph, a contradiction to Lemma \ref{2-cycleStep2}.
\end{proof}


\section[Exactly three chords]{Graphs that do not contain a cycle with exactly three chords as induced subgraph} \label{sec3-cycle}

The aim of this section is to  prove that $\m C_3$ is $\chi$-bounded (Theorem \ref{3-cycleclique4}).

The proof is divided into three parts, according to the clique number. 
First of all, we prove that every (triangle, 3-cycle)-free graph has chromatic number at most $24$. 
Below, the constant $24$ is denoted by $c$. 

For graphs with clique number exactly $3$, we prove that the chromatic number is at most $4c$. When the clique number is at least  $4$, then the chromatic number is close to the clique number. We prove that asymptotically the difference between them is at most one.

Let us state now the exacts statements of these 3 theorems. They are prove in  Subsections \ref{clique2}, \ref{clique3} and \ref{clique4} respectively.

\begin{theorem}\label{3cycleclique2}
A (triangle,\,3-cycle)-free graph  has chromatic number at most $c$.
\end{theorem}

\begin{theorem} \label{lemmaClique3}
A ($K_4$,\,3-cycle)-free graph  has chromatic number at most $4c$.
\end{theorem}

\begin{theorem}\label{3-cycleclique4}
A (3-cycle)-free graph  has chromatic number at most $\max (4c, \omega(G) +1)$.
\end{theorem}

Note that Theorem \ref{3-cycleclique4} says that, if a 3-cycle-free graph has a  large enough clique  (of size at least 96), then $\chi(G) \le \omega(G)+1$.
Moreover the  Haj\'os join of two cliques shows this bound is tight. Let us recall what the Haj\'os join of two cliques is and prove it is 3-cycle-free. 

Let us now describe the construction of the Haj\'os join of two $K_{k}$. 
Take two disjoint copies $H_1$ and $H_2$ of $K_{k-1}$, add a vertex $x$ complete to $H_1$ and $H_2$ and two adjacent vertices $a$ and $b$, such that $a$ is complete to $H_1$ and $b$ is complete to $H_2$. 
The obtained graph is the  Haj\'os join of $K_{k}$ and $K_{k}$ and it is easy to check that it has clique number $k$ and chromatic number $k+1$.
Now, let us show it is 3-cycle-free.
An $ax$-path with interior vertices in $H_1$ is either chordless or has at least two chords. 
Similarly, a $bx$-path with interior vertices in $H_2$ is either chordless or has at least two chords. 
So a cycle going through both $H_1$ and $H_2$ is either chordless, or has exactly two chords, or has more than four chords. 
Hence the graph is 3-cycle-free.


\subsection{Clique number $2$: proof of Theorem~\ref{3cycleclique2}} \label{clique2}

Recall that Theorem~\ref{3cycleclique2} states that \emph{a (triangle,\,3-cycle)-free graph  has chromatic number at most $c=24$}.

To prove this result, we need the  two following  lemmas. 

\begin{lemma}\label{vcycle}
Let $G$ be a (triangle,\,3-cycle)-free graph. For every $z \in V(G)$ and every integer $k$, $S_k(z)$ is  (V-cycle, triangle, 3-cycle)-free.
\end{lemma}

\begin{lemma}\label{crossingcycle}
Let $G$ be a (V-cycle,\,triangle,\,3-cycle)-free graph. For every $z \in V(G)$ and every integer $k$, $S_k(z)$ is (X-cycle, V-cycle,\,triangle,\,3-cycle)-free.
\end{lemma}

Before we prove these two lemmas, let us explain how they imply Theorem~\ref{3cycleclique2}.
Suppose there exists a (triangle,\,3-cycle)-free graph $G$ with $\chi(G) \geq 25$.
Let $z$ be a vertex of $G$. 
By Remark~\ref{scott}, there exists an integer $k$ such that $\chi(G[S_k(z,G)]) \geq 13$.
Put $H= G[S_k(z,G)]$. 
By Lemma~\ref{vcycle}, $H$ is (V-cycle,\,triangle,\,3-cycle)-free.

Let $x$ be a vertex of $H$. 
By Remark~\ref{scott}, there exists an integer $\ell$ such $\chi(G[S_\ell(x,H)]) \ge 7$.
So, by Theorem~\ref{2-cycle}, $G[S_\ell(x,H)]$ contains an  X-cycle as an induced subgraph (it cannot contain a V-cycle since it is an induced subgraph of $H$ that is V-cycle-free) which contradicts~\ref{crossingcycle}.


\subsubsection{Proof of Lemma~\ref{vcycle}}

\textit{Recall that Lemma~\ref{vcycle} states that, if  $G$ is a (triangle, 3-cycle)-free graph and $z$ is a vertex of $G$, then for every integer $k$, $S_k(z,G)$ is (V-cycle, triangle, 3-cycle)-free.}

\medskip

\begin{proof}
Let $G$ be a (triangle, 3-cycle)-free graph and $z$ a vertex of $G$. 
Assume by way of contradiction that there exists an integer $k$ such that $S_k(z,G)$ contains an induced V-cycle $C$. 

\begin{figure} [h] 

\centering
\begin{tikzpicture} 
\draw (2 , 0) node  [below] {$a$};
\draw (1 , 1.7) node [above] {$b$};
\draw (3 , 1.7) node [above] {$c$};

\draw (2 , 0) node {$\bullet$};
\draw (1 , 1.7) node {$\bullet$};
\draw (3 , 1.7) node {$\bullet$};

\draw [dashed] (2 , 0) to[bend left=80] (1 , 1.7);
\draw [dashed] (1 , 1.7) to[bend left=50] (3 , 1.7);
\draw [dashed] (3 , 1.7) to[bend left=80] (2 , 0);

\draw (2 , 0) -- (1 , 1.7);
\draw (2 , 0) -- (3 , 1.7);

\draw (0.9 , 0.7) node[left]  {$P_{ab}$};
\draw (2 , 2.1) node[above]  {$P_{bc}$};
\draw (3.1 , 0.7) node[right]  {$P_{ca}$};

\draw (2,1) ellipse (-3.5 and 2);

\draw (-1.1 , 1.9) node[left]  {$S_k(z,G)$};

\end{tikzpicture}
\caption{The V-cycle $C$ in $S_k(z,G)$.} \label{fig:vcycle}
\end{figure}
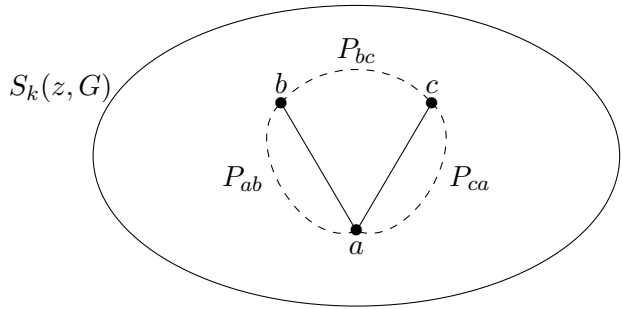

Let $a$ be the unique vertex of $C$ of degree $4$ and let $b,c$ be the two vertices of $C$ of degree $3$. 
Vertices $a$, $b$ and $c$ are called the \emph{important} vertices of $C$.
We denote by $P_{ab}$ the path from $a$ to $b$ contained in $C\backslash \{ ab, ac \}$ that avoids $c$. 
Paths $P_{bc}$ and $P_{ca}$ are defined similarly (see Figure \ref{fig:vcycle}). 
$P_{ab}$, $P_{bc}$ and $P_{ca}$ are called the \textit{intervals} of $C$.
By abuse of notation, $P_{ab}$ will sometimes denote $V(P_{ab})$. 
Also, $P_{ab}$ can be referred to as $P_{ba}$ (and analogously for $P_{bc}$ and $P_{ca}$). 
Moreover, if a vertex $x$ is in $V(P_{ab})$, then the path $xP_{ab}b$ (resp.\ $xP_{ab}a$) can be referred to as $P_{xb}$ or $P_{bx}$ (resp.\ $P_{xa}$ or $P_{ax}$).
Given two vertices $x$ and $y$ in the same interval, the \emph{external path from $x$ to $y$} consists in the path from $x$ to $y$ in $C \sm \{ab, ac\}$ passing through $a, \,b$ and $c$.
If $x_1$ and $x_2$ are two vertices of a path $P$ that have a common father $u$, we say that $x_1$ and $x_2$ are \textit{consecutive neighbors of $u$ along $P$} if $u$ has no neighbors in $\mathring x_1P \mathring x_2$.

Note that adjacent vertices of $C$ cannot have a common father otherwise $G$ would contain a triangle. Also note that if a vertex has 5 neighbors on an induced path, there is a 3-cycle. 

\medskip

The proof consists in studying how a vertex not in $C$ can attach on $C$ and then using unimodal paths to get contradictions.

\medskip

\begin{claim}\label{fact2}
If a vertex $u\notin V(C)$ satisfies $d_C(u) \geq 3$ and all but at most one neighbors of $u$ are contained in an interval of $C$, then $G$ contains a $3$-cycle.
\end{claim}

\begin{proofclaim}
Let $u$ be a vertex not in $V(C)$ such that: $d_C(u) \geq 3$ and all but at most one neighbor of $u$ is contained in an interval of $C$.
So, there exists two vertices   $u_1$ and  $u_2$ in $N_C(u)$ such that $u_1$ and $u_2$ are in the same interval of $C$ and $u$ has exactly one neighbor $u_3$ on the external path $P$ from $u_1$ to $u_2$. 
Then $u_2uu_1Pu_2$ is a 3-cycle with chords $ab,\, ac$ and $uu_3$ (note that, since $G$ is triangle-free, $u_1u_2$ is not an edge).
\end{proofclaim}

\begin{claim} \label{deg3}
If $u \notin V(C)$ and $d_C(u)=3$, then $u$ has exactly one neighbor in each interval i.e.\ it has one neighbor in each of $\mathring P_{ab}$, $\mathring P_{bc}$ and  $\mathring P_{ac}$.
\end{claim}

\begin{proofclaim}
This is immediate by (\ref{fact2}).
\end{proofclaim}

\begin{claim}\label{fact3}
If $u \notin V(C)$, then $u$ has at most $3$ neighbors on $aP_{ab}bP_{bc}c$.
\end{claim}

\begin{proofclaim}
Since $G$ is triangle-free, either $a$ or $b$ is not a neighbor of $u$. If $u$ has at least 5 neighbors in $aP_{ab}bP_{bc}c$, then $u$ has 5 neighbors on one of the chordless paths $\mr aP_{ab}bP_{bc}c$ or $aP_{ab}bP_{bc}\mr c$  which provides a 3-cycle, a contradiction.

So we may assume that $u$ has exactly 4 neighbors in $aP_{ab}bP_{bc}c$ and w.l.o.g.\  that $u$ has at least two neighbors in $P_{ab}$. 
Let $u_1$ and $u_2$ be two consecutive neighbors of $u$ along $P_{ab}$ such that $a$, $u_1$, $u_2$ appear in this order along $P_{ab}$.
Then $u_1 u u_2 P_{u_2b} b P_{bc} c a P_{au_1} u_1$  is a 3-cycle (recall that $u_1u_2$ is not an edge since $G$ is triangle-free).
\end{proofclaim}

Note that by symmetry (\ref{fact3}) also holds for $bP_{bc}cP_{ca}a$.

\medskip

The next claim states the only way a vertex can have at least four neighbors in $C$.

\begin{claim}\label{father}
If $u \notin V(C)$ and $d_C(u) \geq 4$, then $d_C(u)=4$ and $N_C(u)= \{b,c,y_1,y_2\}$ where $y_1$ is the neighbor of $a$ in $P_{ab}$ and $y_2$ is the neighbor of $a$ in $P_{ca}$.
\end{claim}

\begin{proofclaim}
Let $u \notin V(C)$ and suppose $d_C(u) \ge 4$. 

If $u$ has at least three neighbors in $P_{bc}$, then it has at least four neighbors either in $P_{ab}P_{bc}$ or in $P_{bc}P_{ac}$, which contradicts (\ref{father}).
So $u$ has at most two neighbor in $P_{bc}$.

\medskip


\vspace{1.5ex}

\noindent {\bf Case 1} : $u$ has exactly two neighbors, $u_1,u_2$ say, on $P_{bc}$.
\\
Assume w.l.o.g. that $b$, $u_1$, $u_2,c$ appear in this order along $P_{bc}$. By~(\ref{fact3}), $u$ has at most one neighbor on $aP_{ab}\mathring{b}$ and at most one neighbor on $\mathring{c}P_{ca}a$. Since $d_C(u) \geq 4$, both neighbors exists. 
Moreover both are distinct from $a$ otherwise there would be 4 neighbors on $P_{ab}P_{bc}$ or $P_{bc}P_{ca}$, contradicting (\ref{fact3}). 
Denote by $y_1$ (resp. $y_2$) the neighbor of $u$ in $P_{ab}\mathring{b}$ (resp.\ in $\mathring{c}P_{ca}$).

If $u_1 \neq b$ then $y_1 u u_1 P_{u_1c}c P_{ca} P_{ay_1}y_1$ has $3$ chords, namely $uy_2, \, uu_2$ and $ac$. 
So $u_1=b$ and, by symmetry, $u_2=c$. 
If $y_2a$ is not an edge then $y_2 u u_1 P_{u_1a} a c P_{cy_2} y_2$ is a 3-cycle with chords $ab,\,uy_1$ and $uc$. 
So $ay_2$, and by symmetry $ay_1$, are edges, so the outcome holds.

\vspace{1.5ex}

\noindent {\bf Case 2} : $u$ has exactly one neighbor, $u_3$ say on $P_{bc}$.
\\
Since $d_C(u) \geq 4$, $u$ has at least 3 neighbors on $\mathring{b}P_{ba}P_{ac}\mathring{c}$. W.l.o.g. $u$ has at least two neighbors in $P_{ab}\mathring{b}$.
By (\ref{fact3}), $u$ has exactly two neighbors, $u_1, u_2$ say, in $P_{ab} \mathring b$.
Assume that $a$, $u_1$, $u_2$ appear in this order along $P_{ab} \mathring b$.
Let $u_4$ be the neighbor of $u$ that is closest from $a$ in $\mathring P_{ca}$ ($u_4$ exists since $d_C(u) \geq 4$). 
If $u_3 \neq c$ then $u_4 u u_3 P_{u_3b} b P_{ba} P_{ac} u_4$ is a 3-cycle with chords $ab$, $uu_1$ and $uu_2$. 
So $u_3=c$. 

Note that $a$ is not a neighbor of $u$ otherwise $a,c,u$ would be a triangle. So $N_C(u)=\{u_1,u_2,c,u_4\}$, otherwise $u$ would have $5$ neighbors on the chordless path $V(C) \setminus \{a\}$, i.e. there would be 3-cycle. 
Hence $u_2ucP_{ca}P_{au_2} u_2$ is a 3-cycle with chords $uu_1, \, ac$ and $uu_4$, a contradiction.

\vspace{1.5ex}

\noindent {\bf Case 3} : $u$ no neighbor on $P_{bc}$.
\\
Since $d_C(u) \geq 4$, we may assume w.l.o.g.\ that $u$ has at least two neighbors, $u_1,u_2$ say, on $P_{ab}\mathring{b}$. 
By (\ref{fact2}), $u$ has at least two other neighbors $u_3, u_4$ on $\mathring c P_{ca}$. 
Moreover $u$ has no other neighbors in $C$ since otherwise $u$ would have 5 neighbors on the chordless path $V(C) \backslash P_{bc}$. 
By~(\ref{fact2}),  $u_1,u_2,u_3,u_4$ are distinct from $a$. Assume w.l.o.g.\ that $u_2$, $u_1$, $a$, $u_3$, $u_4$ appear in this order along $\mathring{b}P_{ba}P_{ac}\mathring{c}$. 

If $au_3$ is an edge then $u_3 u u_2 P_{u_2a} a c P_{ca}$ is a 3-cycle with chords $au_3, \, uu_4$ and $uu_1$. 
So we may assume $au_3$ is not an edge. 
If $u_2b$ is an edge then $u_2 u u_4 P_{u_4c} c P_{cb} b a P_{au}u_2$ is a 3-ycle with chords $uu_1, \, ac$ and $u_2b$. 
So $u_2b$ is not an edge and hence, $u_2 u u_3 P_{u_3c} c P_{cb} b a P_{au_2} u_2$ is a 3-cycle with chords $uu_1,  \, uu_4$ and $ac$, a contradiction.

\end{proofclaim}

\begin{claim}\label{casmechant}
Let $y$ be a father of a vertex of $C$. Then $d_C(y) \leq 3$. 
\end{claim}

\begin{proofclaim}
 Let $y_1$ be the neighbor of $a$ in $P_{ab}$ and $y_2$ be the neighbor of $a$ in $P_{ca}$.
Let $y$ be the father of a vertex in $C$.
Suppose for contradiction that  $d_C(y) \geq 4$
By~(\ref{father}), $N_C(y)=\{b,c,y_1,y_2\}$. 

Let $e$ be the neighbor of $b$ on $P_{ab}$. 
Note that $e \neq y_1$ since otherwise $aby_1$ is a triangle. 
Let $f$ be a father of $e$. 
By~(\ref{father}), $d_C(f) \leq 3$. 
If $f$ has no neighbor in $P_{ac} \cup P_{ab} \sm \{e,b\}$, then $e f Q_{fy} y c P_{ca}P_{ab} e$ is a 3-cycle with chords $yy_1, \, yy_2$ and $ac$. 
So  $f$ has at least one neighbor in $P_{ac} \cup P_{ab} \sm \{e,b\}$.

Assume that $f$ has at least one neighbor $f_1$ in $P_{ab} \sm \{e,b\}$. 
Then it is the only one, otherwise $d_C(f)=3$ and all neighbors of $f$ are in the same interval contradicting~(\ref{deg3}). 
Then $e f Q_{fy} b a  P_{ae} e$ is a 3-cycle with chords $eb, \, yy_1$ and $ff_1$. 
So $f$ has no neighbors in $P_{ab} \sm \{e,b\}$.

Hence $f$ has at least one neighbor $f_1$ in $P_{ca} \sm \{a\}$ and it is the only one by (\ref{deg3}). 
If $f$ has no neighbor on $\mathring{P_{bc}}$ then $e f Q_{fy} y y_2 P_{ac}  P_{ce}  e$ has chords $yb,\, yc$ and $ff_1$. 
So, $f$ has at least one neighbor $f_2$ in $\mathring P_{bc}$, and it is the only one by (\ref{deg3}).

If $fy$ is an edge then $f_2feP_{ey_1}y_1yy_2P_{y_2c}P_{cf_2}$ is a 3-cycle with chords $fy, \, ff_1$ and $yc$. 
Otherwise $eff_2P_{f_2b}byy_2aP_{ae}$ is a 3-cycle with chords, $eb, \, ab$ and $yy_1$.
\end{proofclaim}

\begin{claim} \label{importantDeg2}
If $x$ is a father of an important vertex of $C$, then $d_C(x) \le 2$.
\end{claim}

\begin{proofclaim}
Let $x$ be a father of an important vertex of $C$. 
By (\ref{casmechant}), $d_C(x) \le 3$. 
By (\ref{deg3}), the father of an important vertex cannot have exactly three neighbors in $C$. 
So $d_C(x) \le 2$.
\end{proofclaim}

\begin{claim} \label{brother}
Let $e$ be the neighbor of $a$ on $P_{ab}$ and let $f$ be a father of $e$. 
Then $d_C(f) \le 2$.
\end{claim}

\begin{proofclaim} 
Assume for contradiction that $d_C(f)=3$.
By~(\ref{casmechant}), $f$ has a exactly one neighbor, $f_1$ say, in $\mathring P_{bc}$ and exactly one, $f_2$ say, in $\mathring P_{ca}$. 
All of them are distinct from important vertices since fathers of important vertices have at most two neighbors on $C$.

Let $x$ be a father of $a$.
If $x$ has no neighbor on $\mathring a P_{ab}P_{bc}$, then  $efQ_{fx}acP_{cb}P_{be} e$ is a $3$-cycle with chords $ab, ae$ and $ff_1$. 
So $x$ has a neighbor, $x_1$ say in $\mathring a P_{ab}P_{bc}$ and, by (\ref{importantDeg2}), $N_C(x)=\{a,x_1\}$.
If $x_1 \in \mathring P_{ab}P_{bf_1}$, 
then $axQ_{xf}f_1P_{f_1b}P_{ba}a$ is a 3-cycle with chords $ab, ef$ and $xx_1$. 
So $x_1 \in f_1P_{f_1c}c$ and then $axQ_{xf}ff_1P_{f_1c}P_{ca}a$ is a 3-cycle with chords $ac$, $ff_2$ and $xx_1$ a contradiction.
\end{proofclaim}


\begin{claim}\label{fatherdeg1}
Let $x$ be a father of $a$. Then $d_C(x)=2$.
\end{claim}

\begin{proofclaim}
By (\ref{importantDeg2}), $d_C(x) \le 2$. 
So we may assume by way of contradiction that $d_C(x)=1$.
Let $e$ be the neighbor of $a$ on $P_{ab}$ and let $f$ be a father of $e$. 
Finally let $y$ be a father of $b$.

If $d_C(f)=1$, then $axP_{xf}eP_{eb}P_{bc}P_{ca}a$ is a 3-cycle with chords $ae$, $ab$, $ac$. 
So $d_C(f) \ge 2$ and thus, by (\ref{brother}),  $d_C(f)=2$. 
If the second neighbor $f_1$ of $f$ is on $P_{ab}P_{bc}$, then $axQ_{xf}eP_{eb}P_{bc}ca$ is a 3-cycle with chords $ae, ab$ and $ff_1$. 
So $f_1 \in \mathring{P}_{ca}$. 

Note that $eb$ is not an edge since otherwise $aeb$ is a triangle. 
If $y$ has a neighbor in $\mathring b P_{bc}P_{ca}e$, then $byQ_{yf}feaP_{ac}P_{cb}b$ is a 3-cycle with chords $ab$, $ac$ and $ff_1$. 
So $y$ has a neighbor, say $y_1$, in $\mathring b P_{bc}P_{ca}e$ and by (\ref{importantDeg2}), $N_C(y)=\{b,y_1\}$.
If $y_1 \neq e$, then $axQ_{xy}bP_{bc}P_{ca}$ is a 3-cycle with chords $ab, ac$ and $yy_1$. 
Hence $y_1=e$ which contradicts (\ref{brother}).
\end{proofclaim}

\medskip

We now have proved enough claims to finish the proof. 
Let $x$ be a father of $a$. 
By (\ref{fatherdeg1}), $d_C(x)=2$. 
Let $x_1$ be the neighbor of $x$ distinct from $a$ on $C$.
 Let $e$ be the neighbor of $a$ on $P_{ab}$ and let $f$ be a father of $e$. 
Finally let $y$ be a father of $b$. 
By symmetry we may assume that $x_1 \in P_{bc}P_{ca} \mathring a$. 

If $d_C(y)=1$, then  $axQ_{xy}bP_{bc}P_{ca}$ would be 3-cycle with chords $ab, ac$ and $xx_1$.
So $d_C(y) \ge 2$ and by (\ref{importantDeg2}), $d_C(y)=2$. 
Let $y_1$ be the neighbor of $y$ distinct from $b$ on $C$.

Assume that both $x_1, y_1$ are on $P_{ca}$. 
If $a, y_1,x_1$ appears in this order along $P_{ac}$ then $x_1 P_{x_1a}P_{ab}yQ_{yx}x_1$ is a 3-cycle with chords $ax$, $ab$ and $yy_1$ ($x_1 \neq c$ since otherwise $acx$ is a triangle). 
So $a,x_1,y_1$ appear in this order along  $P_{ac}$ and $x_1 \neq y_1$. 
In particular $y_1a$ is not an edge and so $axQ_{xy}yy_1P_{y_1c}P_{cb}P_{ba}a$ is a 3-cycle with chords $ab$, $ac$ and $yb$. 
Moreover, if both $x_1, y_1$ are on $P_{bc}$ then $axQ_{xy}bP_{bc}a$ is a 3-cycle with chords $ab, \, xx_1$ and $yy_1$.
So, either $x_1 \in \mathring{P_{bc}}$ and $y_1 \in \mathring{P_{ca}}$ or $x_1 \in \mathring{P_{ca}}$ and  $y_1 \in \mathring{P_{bc}}$.

If $x_1 \in \mathring{P_{bc}}$ and $y_1 \in \mathring{P_{ca}}$ and then $x_1xQ_{xy}y_1P_{y_1a}P_{ab}P_{bx_1}$ is a 3-cycle with chords $ax, \, by$ and $ab$. 
Thus $x_1 \in \mathring{P_{ca}}$ and  $y_1 \in \mathring{P_{bc}}$ and then $x_1xQ_{xy}y_1 P_{y_1b}  P_{ba} P_{bx_1}$  is a 3-cycle with chords $ax, \, by$ and $ab$, a contradiction that put an end to the proof. 
\end{proof}


\subsubsection{Proof of Lemma~\ref{crossingcycle}}

\textit{Recall that Lemma~\ref{crossingcycle} states that, if  $G$ is a (triangle, 3-cycle, V-cycle)-free graph and $z$ is a vertex of $G$, then for every integer $k$, $S_k(z,G)$ is X-cycle-free.}

\begin{proof}
Let $G$ be a (triangle, 3-cycle, V-cycle)-free graph, $z$ a vertex of $G$ and suppose for contradiction that there exists an integer $k$ such that $S_k(z,G)$ contain an X-cycle $C$ as an induced subgraph.

\begin{figure} [h] \label{figcrossing}
\centering
\begin{tikzpicture} 
\draw (0 , 0) node  [below] {$a$};
\draw (0 , 2) node [above] {$b$};
\draw (2 , 2) node [above] {$c$};
\draw (2 , 0) node [above] {$d$};

\draw (0 , 0) node  {$\bullet$};
\draw (0 , 2) node {$\bullet$};
\draw (2 , 2) node {$\bullet$};
\draw (2 , 0) node {$\bullet$};;

\draw [dashed] (0 , 0) to[bend left=50] (0 , 2);
\draw [dashed] (0 , 2) to[bend left=50] (2 , 2);
\draw [dashed] (2 , 2) to[bend left=50] (2 , 0);
\draw [dashed] (2 , 0) to[bend left=50] (0 , 0);

\draw (0 , 0) -- (2 , 2);
\draw (2 , 0) -- (0 , 2);

\draw (-0.4 , 1) node[left]  {$P_{ab}$};
\draw (1 , 2.4) node[above]  {$P_{bc}$};
\draw (2.4 , 1) node[right]  {$P_{cd}$};
\draw (1 , -0.4) node[below]  {$P_{da}$};

\draw (1,1) ellipse (4 and 2.3);

\draw (-2.7 , 1.9) node[left]  {$S_k(z,G)$};
\end{tikzpicture}
\caption{The X-cycle $C$ in $S_k(z,G)$.}
\end{figure}
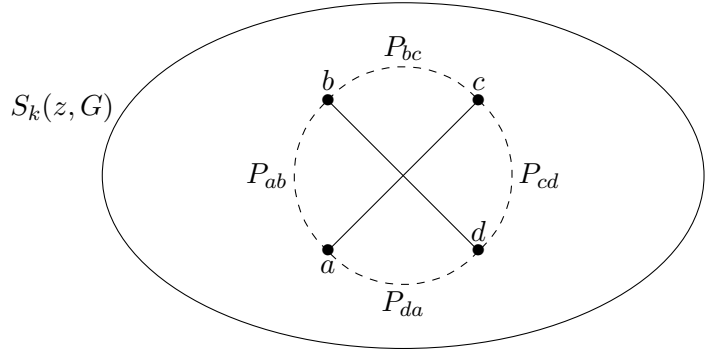

Let $ac$ and $bd$ be the two chords of $C$ and assume that  $a, \, b,\, c, \, d$, appear in this order along $C$.
 Vertices $a, \, b,\, c$, and $d$ are called \emph{important vertices} of $C$. 
Two important vertices that do not form a chord of $C$ are said to be \emph{consecutive}.
An \emph{interval} is an induced path on $C \setminus \{ ac,bd\}$ between two consecutive important vertices. 
They are denoted by $P_{ab}, P_{bc}, P_{cd}$ and $P_{da}$  (see Fig.~\ref{figcrossing}). 
Note that two intervals share at most one vertex.
By abuse of notation, $P_{ab}$ will sometimes denotes $V(P_{ab})$. 
Also, $P_{ab}$ can ve referred to as $P_{ba}$  (and analogously for $P_{cb}$, $P_{dc}$ and $P_{da}$).
Given two vertices $x$ and $y$ in the same interval, the \emph{external path from $x$ to $y$} consists in the path from $x$ to $y$ in $C \sm \{ac, bd\}$ passing through $a, \, b, \, c$ and $d$.

\medskip

The proof is divided in two parts. First we prove that no neighbor of the graph has degree larger than $3$ on $C$. We then study more specifically fathers of important vertices and prove that they are neither of degree $3$, nor $2$ nor $1$.

\begin{claim}\label{fact2c}
If a vertex $u\notin V(C)$ satisfies $d_C(u) \geq 3$ and all but at most one neighbors of $u$ are contained in an interval of $C$, then $G$ contains a $3$-cycle.
\end{claim}

\begin{proofclaim}
Let $u$ be a vertex not in $V(C)$ such that: $d_C(u) \geq 3$ and all but at most one neighbors of $u$ are contained in an interval of $C$.
So, there exists two vertices   $u_1$ and  $u_2$ in $N_C(u)$ such that $u_1$ and $u_2$ are in the same interval of $C$ and $u$ has exactly one neighbor $u_3$ on the external path $P$ from $u_1$ to $u_2$. 
Then $u_2uu_1Pu_2$ is a 3-cycle with chords $ab,\, ac$ and $uu_3$ (note that, since $G$ is triangle-free, $u_1u_2$ is not an edge).
\end{proofclaim}

\begin{claim}\label{father3c}
Every vertex of $u \notin V(C)$, satisfies $d_C(u) \le 3$. Moreover, if $d_C(u)=3$, then no interval contains at least two neighbors of $u$.
\end{claim}

\begin{proofclaim}
Let us first prove a fact. If a vertex $v$ has (at least) 4 neighbors on a path $P$ that has at most one chord, then $G$ contains a V-cycle or a 3-cycle. 
Indeed let $v_1,v_2,v_3,v_4$ be consecutive neighbors of $v$ on $P$. Then $C'=vv_1Pv_4v$ has chords $vv_2$ and $vv_3$. If $v_1Pv_4$  is chordless then $C'$ is a V-cycle otherwise $C'$ is a 3-cycle.

Assume that a vertex $u \notin C$ satisfies $d_C(u) \geq 4$. Since $G$ is triangle free, $u$ is not adjacent to both $a$ and $c$. By symmetry we can assume that $u$ is not adjacent to $a$, then $u$ has 4 neighbors on the path $\mathring{a}P_{ab}P_{bc}P_{cd}P_{da}\mathring{a}$ with at most one chord, contradicting the fact proved above. 

If $d_C(u)=3$, then (\ref{fact2c}) ensures that the neighbors of $u$ are in distinct intervals.
\end{proofclaim}

\medskip

So any vertex $x \notin V(C)$ satisfying $d_C(x)=3$ is adjacent to at most one important vertex. Indeed two important vertices either are opposite (which would create a triangle), or are in a same interval (which would contradict~(\ref{father3c})). 

\medskip

\begin{claim}\label{not2adj1}
Two adjacent vertices of $C$ do not both have a father of degree one on $C$.
\end{claim}

\begin{proofclaim}
Let $u$ and $v$ be two adjacent vertices of $C$. 
Suppose for contradiction that $u$ ( resp.\ $v$) admits a father $u'$ (resp.\ $v'$) such that $d_C(u')=1$ (resp.\ $d_C(v')=1$).
We denote by $P$ the external path from $v$ to $u$. 
Then $uu'Q_{u'v'}vuP$ is a 3-cycle with chords, $uv, \, ac$ and $bd$.
\end{proofclaim}

\begin{claim}\label{no2father3}
Let $x$ be a father of an important vertex. Then $d_C(x) \leq 2$.
\end{claim}

\begin{proofclaim}
Assume by contradiction that a father $x$ of $a$ satisfies $d_C(x)=3$. 
By~(\ref{father3c}), $x$ has exactly one neighbor, $x_1$ say, on $\mathring P_{bc}$, and exactly one neighbor, $x_2$ say, on $\mathring P_{cd}$. Indeed $a$ is in both $P_{ab}$ and $P_{da}$ and~(\ref{father3c}) ensures that there is no two neighbors of $x$ in the same interval.
Note that it implies that neither $bc$ nor $cd$ are edges. 

First assume that $ab$ is an edge. 
Let $y$ be a father of $c$. 
If $d_C(y)=1$, then $axQ_{xy}cP_{cb}dP_{da}$ is a 3-cycle with chords $ab, ac$ and $xx_1$.
So $d_C(y)\ge2$.
If $d_C(y) \ge 3$, then by~(\ref{father3c}), $d_C(y)=3$ and $y$ must have a vertex in $\mathring P_{ab}$ which is impossible since we assumed that $ab$ is an edge.
So $d_C(y)=2$.
If $y_1$ is on $P_{ab}P_{bc}$, then $axQ_{xy}cP_{cb}P_{ba}a$ is a 3-cycle with chords $ac, \,  xx_1$ and $yy_1$,
and if $y_1$ is on $P_{ad}P_{dc}$, then $axQ_{xy}cP_{cd}P_{da}a$ is a 3-cycle with chords $ac, \, xx_2$ and $yy_1$, a contradiction.
So in the following we assume that $ab$, and by symmetry $ad$, are  not edges. 

If $bx_1$ is an edge then $axx_1P_{x_1c}P_{cd}bP_{ba}a$ is a 3-cycle with chords $bx_1, \, xx_2$ and $ac$. 
So $bx_1$, and by symmetry $dx_2$, are not edges.

Let $e$ be the neighbor of $a$ on $P_{ab}$ and $f$ be a father of $e$. 
The cycle $C'=efQ_{fx}x_2P_{x_2c}aP_{ad}bP_{be}e$ has chords $xa$ and $xe$ so it must admit other chords otherwise it is a V-cycle. 
We already showed that $dx_2$ nor $bc$ are edges and that the only neighbors of $x$ in $C$ are $a$, $x_1$ and $x_2$. 
So others chords are due to neighbors of $f$. 
Moreover, $f$ must have at least two neighbors that create chords in $C'$, otherwise $C'$ would be a 3-cycle. 
So, by (\ref{father3c}), $f$ has  one neighbor  on $P_{cd}\mathring{d}$ and one neighbor  on $\mathring{d}P_{da}$.
Let $f_1$ be the neighbor of $f$ in $P_{cd}\mathring{d}$. Then $axQ_{xf}eP_{eb}dP_{dc}a$ is a 3-cycle with chords $ae, xx_2$ and $ff_1$ (note that $ad$ is not an edge since $f$ has a neighbor on $\mathring{P_{ađ}}$).
\end{proofclaim}

\begin{claim}\label{abfather2}
If $ab$ is an edge, fathers of $a$ have degree exactly two. 
\end{claim}

\begin{proofclaim}
Assume by contradiction that a father $x$ of $a$ satisfies $d_C(x) \neq 2$.
So,  by~(\ref{no2father3}), $d_C(x)=1$ . 
Let $z$ be a father of $c$.
The cycle $C'=axQ_{xz}cP_{cb}dP_{da}$ has chords $ac, \, ab$ which, with no additional chords, provides a V-cycle. By~(\ref{no2father3}), $d_C(z) \leq 2$, so $C'$ has at most one other chord due to neighbors of $z$. 
Moreover, if $cd$ is an edge, it is a chord of $C'$. 
Since $C'$ cannot be a V-cycle nor a 3-cycle, $cd$ is an edge and $z$ has another neighbor $z_1$ on $C'$. 
Since both $ab$ and $cd$ are edges, $z_1 \in P_{bc}$ or in $z_1 \in P_{da}$.

Assume first that $z_1 \in {P_{bc}}$. 
If $z_1 \neq b$, then $axQ_{xz}z_1P_{z_1c}dP_{da}$ is a V-cycle with chords $zc$ and $ac$. 
So $z_1=b$.
 Let $e$ be the neighbor of $a$ in $P_{ad}$ (note that $e \neq d$ otherwise $abd$ is a triangle) and let $f$ be a father of $e$. 
 Since $d_C(x)=1$, $d_C(f) \geq 2$  by~(\ref{not2adj1}). 
Such a neighbor, called $f_1$, is unique, otherwise, since $ab$ and $cd$ are edges,  at least two neighbors of $f$ would be in the same interval, contradicting~(\ref{father3c}). 
Note that $f_1 \neq a$ otherwise $aef_1$ is a triangle. 
Then $e f Q_{fz} b P_{bc} d P_{de}$ is a 3-cycle with chords $ff_1, \,zc$ and $bd$. 
So $z_1 \notin P_{bc}$.

Thus, $z_1 \in P_{ad}$. 
Note that $z_1 \neq a$ since otherwise $acz$ is a triangle. 
If $az_1$ is not an edge then $axQ_{xz}z_1P_{z_1d}cP_{cb}a$ is a 3-cycle with chords $ac, \,bd$ and $cz$. 
So $az_1$ is an edge. 
Let $y$ be a father of $b$. 
We have $d_C(y) \leq 2$ by (\ref{no2father3}). 
Then $d_C(y) =2$ by~(\ref{not2adj1}) since $d_C(x)=1$ and $ab$ is an edge. 
Moreover, $ay$ is not an edge otherwise $aby$ is a triangle. 
So $y$ has a neighbor $y_1$ in $\mathring b P_{bc} d P_{dz_1}$.
Therefore $byQ_{yz}z_1P_{z_1d}cP_{cb}$ is a 3-cycle with chords $zc, \, bd$ and $yy_1$.
\end{proofclaim}

\begin{claim}\label{oneimpliesall}
If an important vertex has a father of degree one on $C$, then every father of every important vertex has degree one on $C$.
\end{claim}

\begin{proofclaim}
Assume w.l.o.g.\ that a father $x$ of $a$ satisfies $d_C(x)=1$. 
We show that it implies that every father of $b$ are of degree one in $C$ which, by symmetry, prove the claim.

Let $y$ be father of $b$ and assume for contradiction that  $d_C(y) \neq 1$. 
Note that by~(\ref{abfather2}), neither $ab$ nor $ad$ are edges. 
By~(\ref{no2father3}), $d_C(y)=2$.
Let $y_1$ be the neighbor of $y$ on $C$ distinct from $b$.
Let $b_1$ be the element of $\{ y_1,b\}$ that is nearest from $a$ in $P_{ab}$ and which is distinct from $a$ and let $b_2$ the other one. 
Such a vertex exists since $b$ satisfies the conditions. 
If $ab_1$ is not an edge, then $axQ_{xy}b_1 P_{b_1b}P_{bc}P_{cd}P_{da}$ has $3$ chords $ac,bd$ and $yb_2$. 
So we may assume that $ab_1$ is an edge, since $ab$ is not an edge, $b_1=y_1$. 

Let $z$ be a father of $d$. The cycle $C'=dzQ_{zy}y_1P_{y_1b}P_{bc}P_{cd}$ is, with no additionnal chord, a V-cycle with chords $yb,bd$. Since $d_C(z) \leq 2$ by~(\ref{no2father3}), there is at most one chord with extremity $z$, which provides a 3-cycle.
\end{proofclaim}

\begin{claim}\label{degree1ok}
Fathers of important vertices have degree exactly two on $C$.
\end{claim}

\begin{proofclaim}
Let us prove it by contradiction. By~(\ref{oneimpliesall}), we can assume that all fathers of all important vertices have degree one on $C$. And by (\ref{abfather2}), none of $ab,bc,cd$ and $da$ are edges. Let $u$ be a neighbor of $a$ in $P_{ab}$. Let $x$ be a father of $a$ and $y$ be a father of $u$. By (\ref{not2adj1}), $d_C(y) \neq 1$ since $d_C(x)=1$. 

By~(\ref{father3c}), $d_C(y) \leq 3$. If $d_C(y)=3$ then, by (\ref{father3c}) and (\ref{oneimpliesall}), the neighbors of $y$ are in the interior of distinct intervals. 
 Assume that $y$ has a neighbor in  $\mathring P_{cd}$ and in $\mathring P_{da}$.
Let $y_1$ be the neighbor of $y$ on $\mathring P_{cd}$. 
By~(\ref{oneimpliesall}), a father $y'$ of $b$ satisfies $d_C(y')=1$. 
Hence $by'Q_{y'y}y_1P_{y_1d}P_{da}P_{ab}$ has 3 chords: $bd$ and two chords with extremity $y$. 
It is easy to see that, since fathers of every important vertex are of degree one in $C$, we get a contradiction when $d_C(y)=3$.
So $d_C(y)=2$. 

Let $u'$ be the other neighbor of $a$ and let $z$ be a father of $u'$.
Assume first that $u'$ has a father $z$ distinct from  $y$. 
By symmetry with $y$, $d_C(z)=2$. 
So $u'zQ_{zy}uP_{ub}P_{bc}P_{cd}P_{du'}u'$ has $3$ chords: two chords are given by the other neighbors of $y$ and $z$ and the third one is $bd$. 
So $y$ is adjacent to $u'$.

Let $w$ be the neighbor of $c$ in $P_{dc}$ and $w'$ the neighbor of $c$ on $P_{cb}$.
For symmetric reason why $y$ is a father of both $u$ and $u'$, there exists a vertex $f$ that is the father of both $w$ and $w'$. 
and that is of degree 2 in $C$.
Then $w f Q_{fy} u' a P_{ab} P_{bc} w$ is a 3-cycle with chords $fw'$, $yu$ and $ac$ (recall that both $v$ and $w$ are distinct from $d$ since none of $ad, \, cd$ are edges).
\end{proofclaim}

We are now armed to finish the proof!
By~(\ref{degree1ok}), we may assume that fathers of every important vertex have exactly degree $2$ on $C$. 
Let $x$ and $y$ be some fathers of $a$ and $c$ respectively. 
Since $G$ is triangle-free, $x \neq y$. 
Let us denote by $x_1$ and $y_1$ the other neighbors of respectively $x$ and $y$. 
If $x_1$ and $y_1$ are on $P_{ab}P_{bc}$, then $axQ_{xy}cP_{cb}P_{ba}$ is a 3-cycle with chords $ac, xx_1$ and $yy_1$, a contradiction.
So $x_1$ and $y_1$ cannot both be on  $P_{ab}P_{bc}$ and, symmetrically, they cannot be on $P_{cd}P_{da}$.

So, we may assume w.l.o.g.\ that $x_1$  is on $P_{ab}P_{bc}$ and that $y_1$ is on $P_{cd}P_{da}$. 
If $x_1 \in \mathring{P_{ab}}$ then $x_1xQ_{xy}cP_{cd}P_{da}P_{ax_1}$ is a 3-cycle with chords $ac, \, ax$ and $yy_1$. 
Thus $x_1$ is on $P_{bc}$ and by symmetry $y_1$ is on $P_{da}$. 
More generally, we showed that no father of an important vertex has its second neighbor on the interior of and interval adjacent containing it.
If $ab$ and $cd$ are both not edges, then $axP_{xy}cP_{cb}dP_{da}$ is a 3-cycle with chords $ac, xx_1, yy_1$.
So either $ab$ is an edge, or $cd$ is an edge, or both are edges.

Assume w.l.o.g.\ that $ab$ is an edge. 
A father $w$ of $d$ has it second neighbor $w_1$ neither in $P_{cd}$ nor in $P_{da}$ since no father of an important vertex has its second neighbor on the interior of intervals adjacent to it. 
Since $ab$ is an edge, $w_1 \in  P_{bc}$. 
Now, a father $z$ of $b$ has a unique second neighbor $z_1$ on $\mathring{P_{ad}}$ (by applying the first part of the proof on $b, \, d$ instead of $a, \, c$). 
If $a, y_1, z_1$ appears in this order along $P_{ad}$ then $bzQ_{zy}y_1P_{y_1d}P_{dc}P_{cb}$ is a 3-cycle with chords $bd, zz_1$ and $yc$. 
So $a$, $z_1$, $y_1$ appear in this order along $P_{ad}$ and, in particular, $z_1d$ is not an edge.
Symmetrically, $b$, $x_1$, $w_1$ appear in this order along $P_{cb}$ and $cx_1$ is not an edge.
Finnaly $x_1xQ_{xz}z_1P_{z_1a}cP_{cd}bP_{bx_1}$ is a 3-cycle with chords $ab, xa$ and $zc$, a contradiction.
\end{proof}


\subsection{Clique number $3$ : proof of Theorem~\ref{lemmaClique3}} \label{clique3}

\textit{Recall that Theorem~\ref{clique3} ensures that, if  $G$ is a ($K_4$, 3-cycle)-free graph then $\chi(H) \leq 4c$.}

The proof of Theorem~\ref{lemmaClique3} is organized as follows. 
First of all, we prove (see Lemma \ref{lem:dragonorbutter}) that any ($K_4$, 3-cycle)-free graph with chromatic number at least $2c$ contains either an butterfly as an induced subgraph or a dragonfly as an induced subgraph (see Figure \ref{figDragonfly}). 
Note that the proof is based on Theorem \ref{3cycleclique2}.

We then prove that if a graph $G$ is ($K_4$, 3-cycle)-free and $x$ is a vertex of $G$, then for any integer $k$, $S_k(z,G)$ is (dragonfly,butterfly)-free (see Lemmas \ref{lem:nodragonfly} and \ref{lem:butterfly}).

At the very end, we combine these two result to get the proof of Lemma \ref{lemmaClique3}.

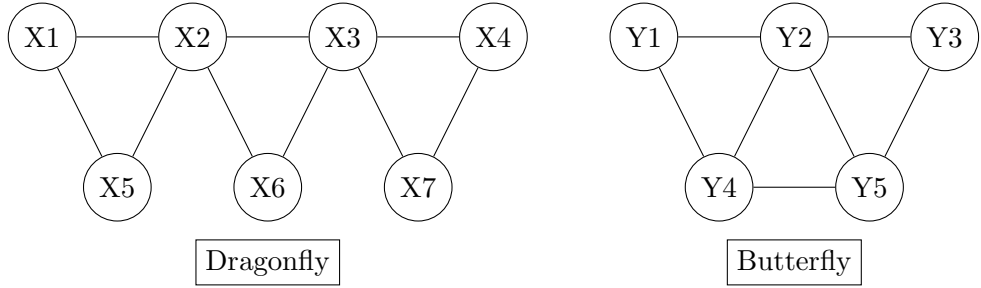
\begin{figure} [h]
\centering
\begin{tikzpicture} 
\node[draw,circle] (X1) at (0,0) {X1}; 
\node[draw,circle] (X2) at (2,0) {X2}; 
\node[draw,circle] (X3) at (4,0) {X3}; 
\node[draw,circle] (X4) at (6,0) {X4}; 
\node[draw,circle] (X5) at (1,-2) {X5}; 
\node[draw,circle] (X6) at (3,-2) {X6}; 
\node[draw,circle] (X7) at (5,-2) {X7}; 

\draw (X1) -- (X2);
\draw (X2) -- (X3);
\draw (X3) -- (X4);
\draw (X5) -- (X1);
\draw (X5) -- (X2);
\draw (X6) -- (X2);
\draw (X6) -- (X3);
\draw (X7) -- (X3);
\draw (X7) -- (X4);

\tikzstyle{quadri}=[rectangle,draw]
\node[quadri] (W) at (3,-3) {Dragonfly};

\node[draw,circle] (Y1) at (8,0) {Y1}; 
\node[draw,circle] (Y2) at (10,0) {Y2}; 
\node[draw,circle] (Y3) at (12,0) {Y3}; 
\node[draw,circle] (Y4) at (9,-2) {Y4}; 
\node[draw,circle] (Y5) at (11,-2) {Y5}; 

\draw (Y1) -- (Y2);
\draw (Y2) -- (Y3);
\draw (Y3) -- (Y5);
\draw (Y4) -- (Y1);
\draw (Y5) -- (Y2);
\draw (Y4) -- (Y2);
\draw (Y4) -- (Y5);

\tikzstyle{quadri}=[rectangle,draw]
\node[quadri] (W) at (10,-3) {Butterfly};
\end{tikzpicture}
\caption{The dragonfly and the butterfly} \label{figDragonfly}
\end{figure} 


\begin{lemma}\label{lem:dragonorbutter}
Let $G$ be a ($K_4$, 3-cycle)-free graph with $\chi(G) > 2c$. 
Then $G$ contains a dragonfly or a butterfly as an induced subgraph.
\end{lemma}

\begin{proof}
All along the proof, the notations of the vertices of dragonfly and butterfly will fit with notations of Figure \ref{figDragonfly}. 
We first prove that $G$ admits a dragonfly or a butterfly as a subgraph. We then prove that it is induced.

\begin{claim} \label{dragonflysubgraph}
$G$ admits a dragonfly as a subgraph.
\end{claim}

\begin{proofclaim}
Let $T \subseteq V(G)$ be a minimal (by inclusion) subset of vertices such that $G \sm T$ is triangle-free.
By Theorem \ref{3cycleclique2}, $G \sm T$ is $c$-colorable.
If $G[T]$ is triangle-free, then $G[T]$ is $c$-colorable and thus $G$ is $2c$-colorable, a contradiction.
Thus, we may assume that $G[T]$ admits a triangle $x_2x_3x_6$.
By minimality  of $T$, $(G \sm T) \cup \{x_2\}$ admits a triangle containing $x_2$, say $x_1x_2x_5$. Similarly, $(G \sm T) \cup \{x_3\}$ contains a triangle containing $b$, say $x_3x_4x_7$.

If $\{x_1,x_5\}=\{x_4,x_7\}$, then $x_1x_2x_3x_5=K_4$, a contradiction.
So $\{x_1,x_5\} \neq \{x_4,x_7\}$.

Assume now that $|\{x_1,x_5\} \cap \{x_4,x_7\}|=1$ and, w.l.o.g., assume that $x_1=x_7$ (see Figure \ref{fig:claim}).
The cycle $C=x_1x_5x_2x_6x_3x_4x_1$ is, if no additional chords exist, a $3$-cycle with chords $x_1x_2$, $x_1x_3$ and $x_2x_3$.
So $C$ must have at least one more chord.
Since $G$ is $K_4$-free, $x_1x_6$, $x_2x_4$ and $x_3x_5$ are not edges.
There remains only 3 possible chords, namely $x_4x_5$, $x_5x_6$ and $x_4x_6$. 
If say  $x_4x_5 \in E(G)$, then $x_1x_2x_5x_4x_3x_1$ is a 3-cycle, a contradiction. 
So $x_4x_5$ is not an edge and, by symmetry, $x_5x_6$ and $x_4x_6$ are not edges.

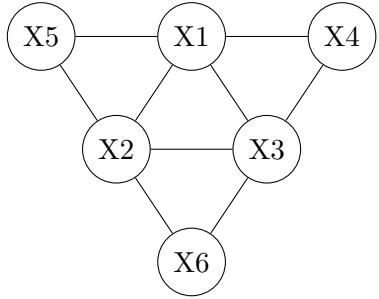
\begin{figure}[h]
\centering
\begin{tikzpicture} \label{fig:claim}
\node[draw,circle] (X1) at (2,1.5) {X1}; 
\node[draw,circle] (X2) at (1,0) {X2}; 
\node[draw,circle] (X3) at (3,0) {X3}; 
\node[draw,circle] (X4) at (4,1.5) {X4}; 
\node[draw,circle] (X5) at (0,1.5) {X5}; 
\node[draw,circle] (X6) at (2,-1.5) {X6}; 

\draw (X1) -- (X2);
\draw (X1) -- (X3);
\draw (X2) -- (X3);
\draw (X4) -- (X1);
\draw (X5) -- (X1);
\draw (X5) -- (X2);
\draw (X6) -- (X2);
\draw (X3) -- (X4);
\draw (X6) -- (X3);
\end{tikzpicture}
\caption{Figure in the proof of Claim (\ref{dragonflysubgraph}).}
\end{figure}

So, $\{x_1,x_5\} \cap \{x_4,x_7\} = \emptyset$ and thus $G[\{x_1,x_2,x_3,x_4,x_5,x_6,x_7\}]$ contains a dragonfly as a subgraph.
\end{proofclaim}

\begin{claim} \label{induit}
$G$ contains either a dragonfly or a butterfly as an induced subgraph.
\end{claim}

\begin{proofclaim}
Observe first that, if $G$ contains a butterfly (see Figure \ref{figDragonfly} for the name of its vertices), then it is induced.
Indeed, since $G$ is $K_4$-free, if it is not induced then $y_1y_3 \in E(G)$ and then $y_1y_4y_5y_3y_2y_1$ is a 3-cycle,  a contradiction.

By (\ref{dragonflysubgraph}), $G$ admits a dragonfly as a subgraph, name it $H$ and refer to Figure \ref{figDragonfly} for the name of its vertices. 
We may assume that $H$ is not induced, otherwise we are done.
If there exists an edge with one extremity in $\{x_1,x_5\}$ and the other one in $\{x_3,x_6\}$, then $G[\{x_1,x_2,x_3,x_5,x_6\}]$ contains a butterfly as a subgraph and thus as an induced subgraph.
So there is no edges with one extremity in $\{x_1,x_5\}$ and the other one in $\{x_3,x_6\}$ and, by symmetry, there is no edges with extremities in $\{x_4,x_7\}$ and the other one in $\{x_2,x_6\}$.

So there exists some edges with one extremity in $\{x_1,x_5\}$ and one extremity in $\{x_4,x_7\}$. 
By symmetry, we may assume w.l.o.g. that $x_5x_7 \in E(G)$.
If it is the only one then, $x_1x_2x_6x_3x_4x_7x_5x_1$ is a 3-cycle, a contradiction.
So it is not the only one and thus, some of $x_1x_4$, $x_1x_7$ or $x_4x_5$ are edges of $G$.
If there is exactly one more, then in the three cases $x_1x_2x_3x_4x_7x_5x_1$ is a 3-cycle, a contradiction.
So, there is at least two more and there is actually exactly two more, otherwise $x_1x_4x_5x_7=K_4$. 
By symmetry between $x_1x_7$ and $x_4x_5$, we may assume w.l.o.g.\ that $x_4x_5 \in E(G)$.  
So one of the edges $x_1x_4$ or $x_1x_7$ exists, but in both cases the cycle $x_2x_6x_3x_4x_7x_5x_2$ is a 3-cycle, a contradiction.
\end{proofclaim}
\end{proof}


\begin{lemma}\label{lem:nodragonfly}
Let $G$ be a ($K_4$, 3-cycle)-free graph and let $z$ be a vertex of $H$. Then, for every integer $i$, $G[S_i(z)]$ is dragonfly-free.
\end{lemma}

\begin{proof}
Assume by way of contradiction that there exists an integer $i$ such that $G[S_i(z)]$ contains a dragonfly as an induced subgraph. 
Name it $H$ and refer to Figure \ref{figDragonfly} for the name of its vertices. 
Let $u$ be a father of $x_5$ and $v$ be a father of $x_7$.

The two next claims examine what are the possible neighborhoods of $u$ and $v$ in $H$.

\begin{claim} \label{x_5}
$N_H(u) \in \{ \{x_5\}, \{x_5,x_1\}, \{x_5,x_2\}, \{x_5,x_3\}, \{x_5,x_6\}, \{x_1,x_3,x_5,x_6\} \}$.
\end{claim}

\begin{proofclaim}
First note that $u$ cannot have exactly two neighbors in $\{x_1,x_2,x_3,x_6\}$. 
Indeed, $u$ cannot see both $x_1$ and $x_2$ since otherwise there is a $K_4$. 
Thus w.l.o.g. $x_6$ is a neighbor of $u$ and then $ux_5x_1x_2x_3x_6u$ is a 3-cycle.

Assume now that $ux_2$ is an edge. Since $G$ is $K_4$-free, $ux_1$ is not an edge. Both $ux_3,ux_6$ are not edges since $G$ is $K_4$-free. 
So none of $ux_3, \, ux_6$ is an edge since otherwise $u$ has exactly  two neighbors in $\{ x_1,x_2,x_3,x_6\}$. If $ux_7$ is an edge then $ux_5x_1x_2x_6x_3x_7u$ is a 3-cycle. So $ux_7$ and by symmetry $ux_4$ are not edges. 
So if $ux_2$ is an edge, then $N_H(u)= \{x_5,x_2\}$ and one of the outcome holds. 
So, we may assume from now on that $ux_2$ is not an edge.

Assume that $ux_7$ is an edge. By symmetry between $x_5$, $x_2$ and $x_7$, $x_3$, we can assume that $ux_3$ is not an edge. 
Let $S=\{x_1,x_4,x_6\}$. 
If $u$ has no neighbor on $S$ then $ux_5x_1x_2x_6x_3x_4x_7u$ is a 3-cycle. 
If $u$ has exactly one neighbor in $S$, then by symmetry between $x_1$ and $x_6$ we may assume that $ux_6$ is not an edge and thus $ux_5x_1x_2x_6x_3x_7u$ is a 3-cycle. 
So $u$ has at least two neighbors in $S$. 
If $u$ has three neighbors in $S$, then $u$ has exactly two neighbors on $\{x_1,x_2,x_3,x_6\}$, a contradiction.
So $u$ has exactly two neighbors in $S$. 
If $ux_1$ and $ux_4$ are edges, then  $ux_5x_1x_2x_6x_3x_7u$ is a 3-cycle.
So, by symmetry between $x_1$ and $x_4$, we may assume that the two neighbors of $u$ in $S$ are $x_4$ and $x_6$.
So  then  $ux_5x_1x_2x_6x_3x_7u$ is a 3-cycle, a contradiction.
So we may assume that $ux_7$,  and by symmetry $ux_4$ are not an edge..

So, $N_H(u) \subseteq \{x_5, x_1,x_3,x_6\}$ and, since we already proved that $u$ does not have exactly two neighbors in $\{x_1,x_2,x_3,x_6\}$, one of the outcome holds.
\end{proofclaim}

\begin{claim} \label{x_7}
$N_H(v) \in \{ \{x_7\}, \{x_4,x_7\}, \{x_3,x_7\}, \{x_2,x_7\}, \{x_6,x_7\}, \{x_2,x_4,x_6,x_7\} \}$.
\end{claim}

\begin{proofclaim}
By obvious symmetries in $H$, the proof is the same as the proof of (\ref{x_5}).
\end{proofclaim}

Note that by claims (\ref{x_5}) and  (\ref{x_7}), $u \neq v$.
In the rest of the proof we show that, whatever the neighborhoods of $u$ and $v$ are, one can find a 3-cycle in $H \cup Q_{uv}$ (recall that $Q_{uv}$ denote a unimodal path linking $u$ and $v$).

\medskip

Suppose first that $d_H(u) \leq 2$ and $d_H(v) \leq 2$.
\\
If $d_H(u)=d_H(v)=1$, then $ux_5x_1x_2x_6x_3x_4x_7vQ_{vu}u$ is a 3-cycle. 
So we may assume that $d_H(u)=2$ and thus, by  (\ref{x_5}),  $u$ has exactly one neighbors in $\{x_1,x_3,x_6\}$. 
Note that by (\ref{x_7}), $vx_1$ is not en edge.
If $d_H(v)=1$, then $ux_5x_1x_2x_6x_3x_7vQ_{uv}u$ is a 3-cycle. 
Moreover, it is still a 3-cycle if $vx_4$ is an edge. 
So $d_H(v)=2$ and $vx_4$ is not an edge. 
Similarly, if $ux_1$ is an edge, then $vx_7x_4x_3x_6x_2x_5uQ_{uv}v$ is a 3-cycle. 
So $ux_1$ is not an edge. 
Now, by (\ref{x_5}) and  (\ref{x_7}), both $u$ and $v$ has exactly one neighbor in $\{x_2,x_3,x_6\}$ and thus $ux_5x_2x_6x_3x_7vQ_{uv}u$ is a 3-cycle, a contradiction.

\vspace{1.5ex}

So, from now on, we assume that $d_H(u)$ and $d_H(v)$ are not both inferior to $2$. 
Hence we may assume  w.l.o.g.\ that  $d_H(u) > 2$, and thus, by~(\ref{x_5}), $N_H(u)=\{x_1,x_3,x_5,x_6\}$. 

Recall that by~(\ref{x_7}), $vx_1$ is not an edge. 
If $v$ has no neighbor in $\{x_2,x_3\}$, then $x_5x_1x_2x_3x_7vQ_{vu}u$ is a 3-cycle. 
So $v$ has at least one neighbor in $\{x_1,x_2,x_3\}$ and by~(\ref{x_7}). 
If $N_H(v) \subseteq \{ \{x_2,x_7\}, \{x_3,x_7\}\}$, then $ux_5x_2x_3x_4x_7vQ_{uv}u$ is a 3-cycle. 
So by~(\ref{x_7}), $N_H(v) = \{x_2,x_4,x_6,x_7\}$

If $uv$ is not an edge, then $ux_5x_2vx_7x_4x_3u$ is a 3-cycle with chords $x_2x_3$, $x_3x_7$ and $vx_4$, a contradiction.
So we may assume that $uv$ is an edge. 
Let $u'$ and $v'$ be fathers of respectively $u$ and $v$ and note that, since $u'$ and $v'$ are in $S_{i-2}(z)$, they have no neighbors in $H$. 
Therefore $u'ux_5x_2x_3x_7vv'Q_{v'u'}u'$ is a 3-cycle, with chords $uv$, $ux_3$ and $vx_2$, a contradiction.
\end{proof}


\begin{lemma}\label{lem:butterfly}
Let $G$ be a ($K_4$, 3-cycle)-free graph and let $z$ be a vertex of $H$. Then, for every integer $i$, $G[S_i(z)]$ is butterfly-free.
\end{lemma}

\begin{proof}
Assume by way of contradiction that there exists an integer $i$ such that $G[S_i(z)]$ contains a butterfly as an induced subgraph. 
Name it $H$ and refer to Figure \ref{figDragonfly} for the name of its vertices. 
Let $u$ be a father of $y_4$ and $v$ be a father of $y_5$.

The two next claims examine what are the possible neighborhoods of $u$ and $v$ in $H$.

\begin{claim} \label{y_4}
$N_H(u) \in \{ \{y_4\}, \{y_2,y_4\}, \{y_1,y_3,y_4\} \}$
\end{claim}

\begin{proofclaim}
Assume first that $|N_H(u)|=2$.
If $N_H(u) = \{y_1,y_4\}$, then $uy_1y_2y_3y_5y_4u$ is a 3-cycle, a contradiction.
If $N_H(u) = \{y_3,y_4\}$, then $uy_4y_1y_2y_5y_3u$ is a 3-cycle, a contradiction.
If $N_H(u) = \{y_4,y_5\}$, then $uy_4y_1y_2y_3y_5u$ is a 3-cycle, a contradiction.
So, if $|N_H(u)|=2$, then $N_H(u)= \{y_2,y_4\}$ and one of the outcome of the theorem holds.

Assume now that $|N_H(u)|=3$. 
If $N_H(u)=\{y_2,y_3,y_4\}$, then $uy_4y_2y_5y_3u$ is a 3-cycle, a contradiction.
If $N_H(u)=\{y_1,y_3,y_4\}$, then one of the outcome of the theorem holds.
So, since $G$ is $K_4$-free, $u$ has to see $y_5$.
The third neighbor of $u$ is thus $y_1$ or $y_3$ and, by symmetry, we may assume that it is $y_1$.
Therefore $uy_4y_1y_2y_5u$ is a 3-cycle, a contradiction.

So we may assume that $|N_H(u)| \ge 4$. 
Since $G$ is $K_4$-free,  $|N_H(u)| = 4$ and $N_H(u) =\{y_1,y_3,y_4,y_5\}$. Thus $uy_1y_4y_2y_5u$ is a 3-cycle, a contradiction.
\end{proofclaim}

\begin{claim}\label{y_5}
$N_H(v) \in \{ \{y_5\}, \{y_2,y_5\}, \{y_1,y_3,y_5\} \}$
\end{claim}

\begin{proofclaim}
By obvious symmetries in $H$, the proof is the same as the proof of (\ref{y_4}).
\end{proofclaim}

\medskip

Note that by claims (\ref{y_4}) and  (\ref{y_5}), $u \neq v$.
In the rest of the proof we show that, whatever the neighborhoods of $u$ and $v$ are, one can find a 3-cycle in $H \cup Q_{uv}$ (recall that $Q_{uv}$ denote a unimodal path linking $u$ and $v$).

\vspace{1.5ex}

\noindent {\bf Case 1} : $N_H(v) = \{y_5\}$.
\\
If $N_H(u)=\{y_4\}$ then $uy_4y_1y_2y_3y_5vQ_{vu}u$ is a 3-cycle, 
a contradiction.
So  $N_H(u) \in \{ \{y_2,y_4\}, \{y_1,y_3,y_4\} \}$ and thus $uy_4y_2y_3y_5vQ_{vu}u$ is a 3-cycle, a contradiction. This completes the proof in case 1.

\vspace{1.5ex}

So from now on, we may assume that $N_H(v) \neq \{y_5\}$ and, by symmetry, that $N_H(u) \neq \{y_4\}$.

\vspace{1.5ex}

\noindent {\bf Case 2} : $N_H(v) = \{y_2,y_5\}$.
\\
If $N_H(u)=\{y_2,y_4\}$ then $uy_4y_2y_5vQ_{vu}u$ is a 3-cycle. Otherwise, we may assume that $N_H(u)=\{y_1,y_3,y_4\}$ and then $uy_1y_2y_3y_5vQ_{vu}$ is a 3-cycle, a contradiction.

\vspace{1.5ex}

So from now on, we may assume that $N_H(v) \neq \{y_2,y_5\}$ and by symmetry, $N_H(u) \neq \{y_2,y_4\}$. 
This leads to the following last case.

\vspace{1.5ex}

\noindent {\bf Case 3} : $N_H(v) = \{y_1, y_3,y_5\}$ and $N_H(u)=\{y_1,y_3,y_4\}$.
\\
If $uv$ is not an edge then $uy_1y_4y_5vy_3u$ is a 3-cycle, with chords $uy_4$, $vy_1$ and $y_3y_5$, a contradiction.
So $uv$ is an edge.
Let $u'$ and $v'$ be fathers of respectively $u$ and $v$. If $u'$ (or $v'$) is adjacent to both $u$ and $v$ we assume that $u'=v'$.
Note that since $u'$ and $v'$ are in $S_{k-2}$ they have no neighbors in $H$.
Therefore $u'uy_1y_2y_3vv'Q_{u'v'}u'$ is a 3-cycle, with chords $uv$, $uy_3$ and $vy_1$, a contradiction.
This completes the proof in Case 3.

\end{proof}

We can now give the proof of Theorem~\ref{lemmaClique3} recall that it states that every ($K_4$,3-cycle)-free graph  has chromatic number at most $4c$.

\medskip

\begin{proof}
Assume by contradiction that there exists a ($K_4$, $3$-cycle)-free graph $G$ that satisfies $\chi(G) \geq 4c+1$ and let $z$ be a vertex of $G$. 
By Remark~\ref{scott}, there exists an integer $k$ such that $S_k(z,G)$ has  chromatic number at least $2c+1$. 
So, by Lemma \ref{lem:dragonorbutter}, it must contain a dragonfly or a butterfly, which is a contradiction with Lemma \ref{lem:nodragonfly} or Lemma  \ref{lem:butterfly}.
\end{proof}


\subsection{Clique number at least $4$ : proof of Theorem~\ref{3-cycleclique4}} \label{clique4}

\textit{Recall that Theorem \ref{lemmaClique3} states that every (3-cycle)-free graph  has chromatic number at most $\max (4c, \omega(G) +1)$.}

\medskip

\begin{proof}
Consider by contradiction the smallest (in number of vertices) graph $G \in \m C_3$ such that $\chi (G) > \max (\omega(G) +1, 4c)$. By Theorem~\ref{3cycleclique2} and~\ref{lemmaClique3}, we have $\omega(G) \ge 4$.
Put $\omega(G)=\omega$. Let $K$ be a largest clique of $G$ and denote by $x_1, \dots, x_{\omega}$ the vertices of $K$.

\begin{claim} \label{mindeg}
Every vertex of $G$ is of degree at least $\omega +1$.
\end{claim}

\begin{proofclaim}
If a vertex $v$ of $G$ is of degree at most $\omega$, then by minimality of $G$ we can color $G \sm \{v\}$ with $max (\omega(G) +1, 4c)$ colors and extend the coloring to $G$, a contradiction.
\end{proofclaim}

\begin{claim} \label{cliquecutset}
$G$ does not admit clique cutsets.
\end{claim}

\begin{proofclaim}
Assume by contradiction that $G$ has a clique cutset A. 
Let $C_1$ be a connected component of $G \sm A$, and $C_2$ the union of all others components.
By minimality of $G$, we may color $G[C_i \cup K]$ with $max (\omega(G) +1, 4c)$ colors for $i=1,2$. 
By using the same colors for the vertices of $A$ in the coloring of $G[C_1 \cup K]$ and $G[C_2 \cup K]$, we can extend the coloring to $G$, a contradiction.
\end{proofclaim}

\begin{claim} \label{degdansK}
If $u \in N(K)$, then $d_K(u) = 1$ or $\omega -1$.
\end{claim}

\begin{proofclaim}
Assume by way of contradiction that $u$ has at least two neighbors in $K$, say $x_1$ and $x_2$, and at least two non-neighbors, say $x_3$ and $x_4$. Then $ux_1x_3x_4x_2u$ is a 3-cycle, with chords $x_1x_2$, $x_1x_4$ and $x_2x_3$, a contradiction.
\end{proofclaim}

Define $S_i=\{u \in N(K) | N_K(u)=\{x_i\}\}$, $T_i=\{u \in N(K) | N_K(u)=V(K) \sm \{x_i\}\}$
and, for all $i = 1, \dots, \omega$, $U_i=S_i \cup T_i$.

An $uv$-path $P$ is an \textit{N(K)-connection} if no vertex of $P$ is in $K$ and $N(K) \cap P = \{u,v\}$. Note that vertices of $\mathring{P}$ have no neighbors on $K$ and that an $N(K)$-connection can be  an edge.

\begin{claim} \label{connection}
Let $P$ be an $N(K)$-connection with endvertices $u$ and $v$. 
Then there exists an integer $i$ such that $\{u,v\} \subseteq U_i$ and $\{u,v\} \not\subseteq T_i$.
\end{claim}

\begin{proofclaim}
Let $i$, $j$, $k$ and $l$ be 4 distinct integers in $\{1, \dots, \omega \}$. Such integers exist since $\omega\ge 4$.
\\
If $u \in T_i$ and $v \in T_j$, then $ux_jx_kx_ivPu$ is a 3-cycle, with chords $ux_k$, $vx_k$ and $x_ix_j$.
\\
If $u \in S_i$ and $v \in T_j$, then $ux_ix_kx_lvPu$ is a 3-cycle, with chords $vx_i$, $vx_k$ and $x_ix_l$.
\\
If $u \in S_i$ and $v \in S_j$, then $ux_ix_kx_lx_jvPu$  is a 3-cycle, with chords $x_ix_j$, $x_ix_l$ and $x_jx_k$.
\\
If $u \in T_i$ and $v \in T_i$, then $ux_jx_ix_kvPu$ is a 3-cycle, with chords $ux_j$, $vx_k$ and $x_jx_k$.
\end{proofclaim}

\begin{claim} \label{U_i}
There is a unique $i \in \{1, \dots, \omega\}$ for which $U_i \neq \emptyset$.
\end{claim}

\begin{proofclaim}
Let us argue by way of contradiction.
By (\ref{cliquecutset}), $G \sm K$ is connected, so there exists a  path $P$ in $G \sm K$ from $U_i$ to $U_j$ such that $i \neq j$. Choose $P$ subject to its minimality. It is clear that $P$ is an $N(K)$-connection and thus it contradicts (\ref{connection}).
\end{proofclaim}

By (\ref{U_i}), we may assume w.l.o.g. that $U_1 \neq \emptyset$ and, for any $i \neq 1$, $U_i = \emptyset$. 
Moreover, $S_1$ and $T_1$ both contain at least two vertices, otherwise $x_1$ or $x_2$ have degree at most $\omega$, a contradiction to (\ref{mindeg}).

\medskip

We say that a vertex $x$ is \emph{complete} to a set of vertex $S$ is $x$ is adjacent to every vertex in $S$.

\begin{claim} \label{complete}
If there exists an $N(K)$-connection from a vertex of $T_1$ to a vertex $s_1 \in  S_1$, then $s_1$ is complete to $T_1$.
\end{claim}

\begin{proofclaim}
Let $P$ be a minimal $N(K)$-connection from $s_1$ to $T_1$. 
Denote by $t_1 \in T_1$ the second endvertex of $P$. 
Assume by way of contradiction that there exists a vertex $t_2 \in T_1 \sm \{t_1\}$ that is not adjacent to $s_1$. 
Then there is no edge linking $t_2$ with a vertex of $P$, otherwise there would be an $N(K)$-connection from $t_1$ to $t_2$, contradicting (\ref{connection}). 
So,  $s_1Pt_1x_2t_2x_3x_1s_1$ is a 3-cycle with chords $x_1x_2$, $t_1x_3$ and $x_2x_3$, a contradiction.

So $s_1$ is complete to $T_1 \setminus \{t_1\}$ and, by minimality of $P$, $s_1$ is adjacent to $t_1$.
\end{proofclaim}

\begin{claim} \label{twin}
$N(T_1) \subseteq S_1 \cup K$.
\end{claim}

\begin{proofclaim}
Assume by contradiction that there exists $t_1 \in T_1$ such that  $N(t_1) \nsubseteq S_1 \cup K$. 
Since $t_1$ is not a cutvertex by~(\ref{cliquecutset}), consider a minimal path $P'$ from $N(t_1) \sm (S_1 \cup K)$ to $N(K)$ in $G \sm \{ t_1 \}$. 
Call $t_1'$ and $x$ the extremities of $P$ with $t_1' \in N(t_1) \sm (S_1 \cup K)$ and put $P=t_1Px$.
Observe that if $t_1x$ is not an edge, then $t_1Px$ is an $N(K)$-connection and that in both cases there exists an $N(K)$-connection linking $t_1$ and $x$. 
So, by (\ref{connection}), $x \notin T_1$. 
Hence, by (\ref{complete}), $x$ is complete to $T_1$ and in particular $xt_1$ is an edge.
 Finally $t_1Ps_1x_1x_2x_3t_1$ is a 3-cycle with chords $s_1t_1$, $x_1x_3$ and $t_1x_2$, a contradiction.
\end{proofclaim}

\begin{claim} \label{final}
For any vertex $t \in T_1$, $N(t)=S_1 \cup K \sm \{x_1\}$.
\end{claim}

\begin{proofclaim}
Let $t \in T_1$.
By~(\ref{connection}), $T_1$ is a stable set. 
So if $t$ is not adjacent to any vertex of $S_1$, $N(t)=K \sm \{x_1\}$, a contradiction to (\ref{mindeg}). 
So $t$ is adjacent to at least one vertex in $S_1$ and thus, by~(\ref{complete}), $t$ is complete to $S_1$.
\end{proofclaim}

Let $t_1$ and $t_2$ be two distinct vertices in $T_1$ (remind that they exist because  if $|T_1|=1$, then $d(x_2)=\omega$, contradicting~(\ref{mindeg})). 
By~(\ref{final}), $N(t_1)=N(t_2)=S_1 \cup K \sm \{x_1\}$.
By minimality of $G$, $G \sm \{t_2\}$ admits a proper coloring $\gamma$ with $\max (4c, \omega +1)$ colors. Since $t_1t_2 \notin E(G)$ and $N(t_1)=N(t_2)$, $\gamma$ can be extended to $G$ by giving  to $t_2$ the same color as $t_1$.
\end{proof}

\end{document}